\documentclass[12pt]{article}
\usepackage[margin=0.9in]{geometry}
\usepackage[utf8]{inputenc}
\usepackage{bm}
\usepackage{bbm}
\usepackage{stmaryrd}
\usepackage{amsmath}
\usepackage{amsthm}
\usepackage{amssymb}
\usepackage[dvipsnames]{xcolor}
\usepackage{graphicx}
\usepackage{ulem}
\usepackage{url}
\usepackage{comment}
\usepackage{braket}

\newcommand{\s}{{\rm S}}
\newcommand{\rr}{{\rm R}}

\newcommand{\tr}{\mathrm{Tr}}

\newcommand{\h}{{\mathcal H}}

\newcommand{\ve}{\varepsilon}

\newcommand{\bbbone}{\mathchoice {\rm 1\mskip-4mu l} {\rm 1\mskip-4mu l}
{\rm 1\mskip-4.5mu l} {\rm 1\mskip-5mu l}}
\newtheorem{thm}{Theorem}
\newtheorem{prop}{Proposition}
\newtheorem*{prop*}{Proposition}
\newtheorem{lem}{Lemma}
\newtheorem{cor}{Corollary}

\renewcommand{\a}{{\rm A}}
\renewcommand{\b}{{\rm B}}

\newcommand{\vve}{{\vec\ve}}
\newcommand{\vvp}{{\vec p}}

\newcommand{\bop}{\beta_{\rm c}}

\usepackage{authblk}
\title{The contact temperature of 
arbitrary quantum states}

\author[1]{Alain Joye\footnote{alain.joye@univ-grenoble-alpes.fr}}
\author[2]{Marco Merkli\footnote{merkli@mun.ca}}
\affil[1]{Univ.~Grenoble Alpes

CNRS, Institut Fourier

F-38000 Grenoble France
\medskip
}
\affil[2]{Department of Mathematics and Statistics

Memorial University of Newfoundland

St.~John’s, NL, Canada A1C 5S7
}

\begin{document}

\maketitle

\begin{abstract} 
An intuitive scheme to assign a temperature to an arbitrary state of a quantum system is to investigate the heat flow resulting from the coupling to a thermometer. We introduce a simple model of a universal thermometer with the following property. When it is prepared in a Gibbs equilibrium state at inverse temperature $\beta\in\mathbb R$ and brought into thermal contact with a system in any state, the heat flow between the system and thermometer  vanishes for a unique value of $\beta$.  We call this value the {\it contact temperature} $\bop\in\mathbb R$ of the system state. The thermometer is universal in that it yields a unique contact temperature for arbitrary states of  finite dimensional quantum systems.

\end{abstract}

\section{Introduction}\label{sec:intro}

The equilibrium state of a quantum system  whose dynamics is generated by a Hamiltonian $H$ is given by the Gibbs density matrix 
$$
\gamma_\beta = \frac{e^{-\beta H}}{{\rm Tr}\, e^{-\beta H}}\, ,
$$
where $\beta\in\mathbb R$ is the (inverse) equilibrium temperature. The Gibbs density matrix is obtained as the maximizer of the entropy under the constraint of a fixed mean energy. Different approaches have been proposed to extend the notion of temperature to general quantum states $\rho$, see the review \cite{Casas}.  Among them is the effective temperature $\beta^*$ of $\rho$ defined as the Gibbs equilibrium temperature resulting in the same average energy, determined by ${\rm Tr}(H\rho ) = {\rm Tr}(H \gamma_{\beta^*})$. The effective $\beta^*$ is the unique minimizer of the function $\beta\mapsto S(\rho | \gamma_\beta)$, the von Neumann relative entropy of $\rho$ with respect to the Gibbs state. An approach parallel to standard (equilibrium) thermodynamics is developed in \cite{Alipour}, where a non-equilibrium temperature is constructed as the partial derivative of the von Neumann entropy with respect to an internal energy. In the current work we follow a more concrete, operational strategy for defining the temperature, in line with some previous work \cite{Muschik}. There a contact temperature is defined for classical thermodynamical systems. The idea is to couple the system in a non equilibrium state to a thermostat in equilibrium and to analyze for which thermostat temperature the total heat exchange through the contact boundary vanishes. In the current paper we develop this idea for quantum systems. We assign a temperature to arbitrary states $\rho$ of any (finite dimensional) quantum system by coupling it to a thermometer or thermostat, an (infinite dimensional) quantum system in thermal equilibrium.  The contact with the thermometer is mediated by an interaction producing equal transition probabilities between all system+thermometer states at a fixed total energy. We show that this results in a single value of the thermometer temperature before the coupling, for which the heat flow vanishes. This value is the effective contact temperature of the system. The contact temperature can be positive or negative, as is the case for virtual temperatures in thermal machines \cite{BLPS} and effective temperatures of a related approach \cite{LPB}, where the authors introduce two temperatures $\beta_c$ and $\beta_h$ (cold and hot) associated to a quantum state $\rho$, defined by the ability of the system in the state $\rho$ to either furnish or extract energy when coupled to reservoirs.

\subsection{Outline of approach and of main results}

We now explain our main results in a somewhat  informal way and we present the rigorous aspects in the coming sections. The system, called $\a$, is a $d$-level quantum system with a Hamiltonian $H_\a$ having non-degenerate energies $\ve_1 <\cdots <\ve_d$. The thermometer, or system $\b$, has Hamiltonian $H_\b$ with eigenvalues $\eta_k=k/N$, $k=0,1,\ldots$ and $N$ large. $\a$ and $\b$ are coupled by an energy exchange interaction described by a {\it thermal process} $U$, a unitary acting on the total complex $\a+\b$,  having the commutation property
$$
[U,H_\a+H_\b] =0.
$$
As a consequence of the commutation property, a change in the energy of $\a$ is entirely due to a change in the energy of $\b$ and vice versa. Given a state (density matrix) $\rho$ of $\a$, our main quantity of interest is the change in the energy of $\a$ due to the contact with $\b$ in the Gibbs state
$$
\tau_\beta=\frac{e^{-\beta H_\b}}{{\rm Tr}e^{-\beta H_\b}},
$$
which is given by
$$
Q_\a(\beta) = \tr_{\a\b}{\big(U(\rho\otimes\tau_{\beta}) U^* H_\a\big)} - \tr_\a{(\rho H_\a)}.
$$
The state $\rho\otimes\tau_\beta$ of A+B before the contact is factorized, and so the energy change may be interpreted as heat. We call $Q_\a(\beta)$ the energy (or heat) flow resulting from the contact with the thermometer. The thermometer Gibbs state $\tau_\beta$ is defined only for $\beta>0$ as ${\rm Tr}e^{-\beta H_\b}=\infty$ for $\beta\le 0$, because $\b$ has infinitely many energy levels. Thus in the above definition of $Q_\a(\beta)$ we consider $\beta>0$. As expected, the flow vanishes if $\a$ is in the Gibbs state at the same temperature as $\b$: for $\rho=\gamma_\beta$ we have $Q_\a(\beta)=0$, which follows from the commutation property $[U,H_\a+H_\b]=0$.

In order to guarantee a good thermal contact between $\a$ and $\b$, the two systems should be able to exchange energy. In the theory of open quantum systems, particularly in weak coupling regimes, well-coupledness is often stated as a `Fermi Golden Rule' condition. While the interaction in the present work is not scaled by any coupling constant, our well-coupledness condition shares the same two central physical aspects with the Fermi Golden Rule. Namely, the system and thermometer (reservoir) should have compatible Bohr energies (energy differences) and the interaction should allow energy transitions in tandem, within the system and the reservoir. In the current setup the total energy is conserved by $U$, so that energetic transitions in $\a$ must be accompanied by equal and opposite energetic transitions in $\b$, that is, the Bohr energies of $\a$ should also be Bohr energies of $\b$. Given that the thermometer spectrum is $k/N$, $k=0,1,\ldots$ this leads us to the following assumption.
\smallskip

(i) For all $\ell,m$,  the difference $|\ve_\ell-\ve_m|$ must be a multiple of $1/N$. 
\smallskip

\noindent
In order to allow for energy exchange transitions, the unitary $U$ must have non-vanishing matrix elements between energy eigenstates of $H_\a+H_\b$.  We consider a thermal contact for which all equal energy transitions are {\it equiprobable}. Namely, we assume that
\smallskip

(ii) The transition probabilities induced by $U$, between any two equal energy eigenstates of the uncoupled Hamiltonian $H_\a+H_\b$, are the same. 
\smallskip

\noindent 
The equal probability property (ii) is an expression of {\it randomness}, which we think is natural for complex or `generic' systems. Namely, we show in Section \ref{sec:model} that for the purpose of the  heat flow, the property (ii) is equivalent to taking the unitary $U$ as a random operator, drawn according to the unique, bi-invariant probability measure on the unitary group (the Haar measure). 
Besides its spectrum, the manner how the thermometer couples to systems $\s$ is an equally important part of its definition. In our model this manner is encoded in the operator $U$ which depends on the spectrum of $\s$ due to the condition (ii). Part of the definition of the thermometer, therefore, is that it will couple to any system in a fashion which produces equal probability energy transitions, irrespective of the detailed nature of $\s$.

Even a small change in the system energies can alter the Bohr resonance condition (i) and  simultaneously the structure of $U$. This may lead to a decreased value of the mathematical expression of the heat flow. For example, take energies $\ve_\ell$ such that (i) is satisfied with a resulting $Q_\a>0$. An infinitesimal change of the energies into new ones with the property that for all $\ell\neq m$, the difference  $|\ve'_\ell-\ve'_m|$ is not an exact multiple of $1/N$, implies that $U'$ satisfying condition (ii) with the new energies, reduces to (a direct sum of) pure phases and consequently $Q'_\a=0$. The flow is thus not continuous in the system energies, when one only considers (ii) without (i).  Physically, however, a drastic change in the heat flow under infinitesimal alteration of the energy spectrum of $\s$ is unrealistic. The conditions (i) and (ii) work together on a physical level. Together they  eliminate this unnatural instability. Still, given an arbitrary energy spectrum $\{\ve_\ell\}$ of $\s$, the rationality condition (i) will not be verified in general. However, for $N$ large enough, the condition is satisfied to an arbitrary degree of precision. It turns out, as we show below, that the heat flow $Q_\a(\beta)$ {\it does not depend on $N$}. Therefore the condition (i) is not a physically relevant restriction on $\s$.

In conlcusion, jointly the assumptions (i) and (ii) give a universal thermometer model, yielding a unique contact temperature associated to any state $\rho$ of $\s$, regardless of the physical nature of $\s$.

{\bf Thermodynamic results.} Our first result is Theorem \ref{Qunbounded}, in which we obtain the following explicit expression for the heat flow, 
$$
Q_\a(\beta)=\sum_{1\leq j < k \leq d}\frac{e^{-\beta (\ve_k -\ve_j)}}{k(k-1)}p_j \sum_{l=1}^{k-1}(\ve_k-\ve_l)- \sum_{k=2}^{d}\frac{p_k}{k}\sum_{l=1}^{k-1}(\ve_k-\ve_l),
$$
where the $p_j$ are the ($H_\a$ energy) populations of the state $\rho$. The explicit form of the heat flow allows for a direct discussion of its properties (see Section \ref{sec:model}). Our main question now is whether there is a unique root (zero) $\bop$ of the function $Q_\a(\beta)$. We show that (unless $\rho$ is the ground state or the top excited state of $H_\a$), $\beta\mapsto Q(\beta)$ is continuous and strictly decreasing in $\beta$ (for all $\beta\in\mathbb R$) and that $Q(\beta)<0$ for $\beta\rightarrow\infty$ while
$$
Q_\a(0) :=\lim_{\beta\rightarrow 0_+} Q_\a(\beta) =  \sum_{j=1}^d \ve_j(\frac1d -p_j).
$$
Consequently, if $Q_\a(0)\ge 0$, then there exists a unique $\bop\ge 0$ which makes the flow vanish. The root $\bop$ is called the contact temperature of $\rho$. The existence and uniqueness of $\bop\ge 0$ is the content of our Theorem \ref{thm2}.

As is readily visible from the explicit expression above, $Q_\a(0)$ is negative for very legitimate states $\rho$ (populations $p_j$). In this case the monotonicity implies that the function $Q_\a(\beta)$ does not have any non-negative root. How should we assign a temperature to states for which $Q_\a(0)<0$ ?  The physical meaning of this inequality is that the system $\a$ in the state $\rho$ will transfer heat to $\b$ even if $\b$ is infinitely hot, $\beta=0_+$ (measured according to the Gibbs equilibrium temperature $T=1/\beta$). In order to define the contact temperature in the case when $Q_\a(0)<0$, the first impulse is to use the above explicit formula for $Q_\a(\beta)$ for all values of $\beta\in\mathbb R$ | even though it was derived for $\beta>0$. The monotonicity and the boundary values $\lim_{\beta\rightarrow-\infty}Q_\a(\beta)=\infty$ and $\lim_{\beta\rightarrow\-\infty}Q_\a(\beta)<0$ would then guarantee the existence of a single root in $(-\infty,0)$ for $Q_\a(0)<0$. However, this extension of $Q_\a(\beta)$ to negative $\beta$ is not physical. Indeed, $Q_\a(\beta)$ is supposed to represent the energy difference of a finite system before and after the contact with the thermometer; so $Q_\a(\beta)$ must be bounded above by $\ve_d-\ve_1$ and cannot have the above stated boundary condition at $\beta=-\infty$. The origin of this inconsistency is that the Gibbs state $\tau_\beta$ is ill defined for $\beta\le 0$, and therefore so is the defining expression of $Q_\a(\beta)$.

\begin{figure}[h!]
\centering
\includegraphics[width=.9\linewidth]{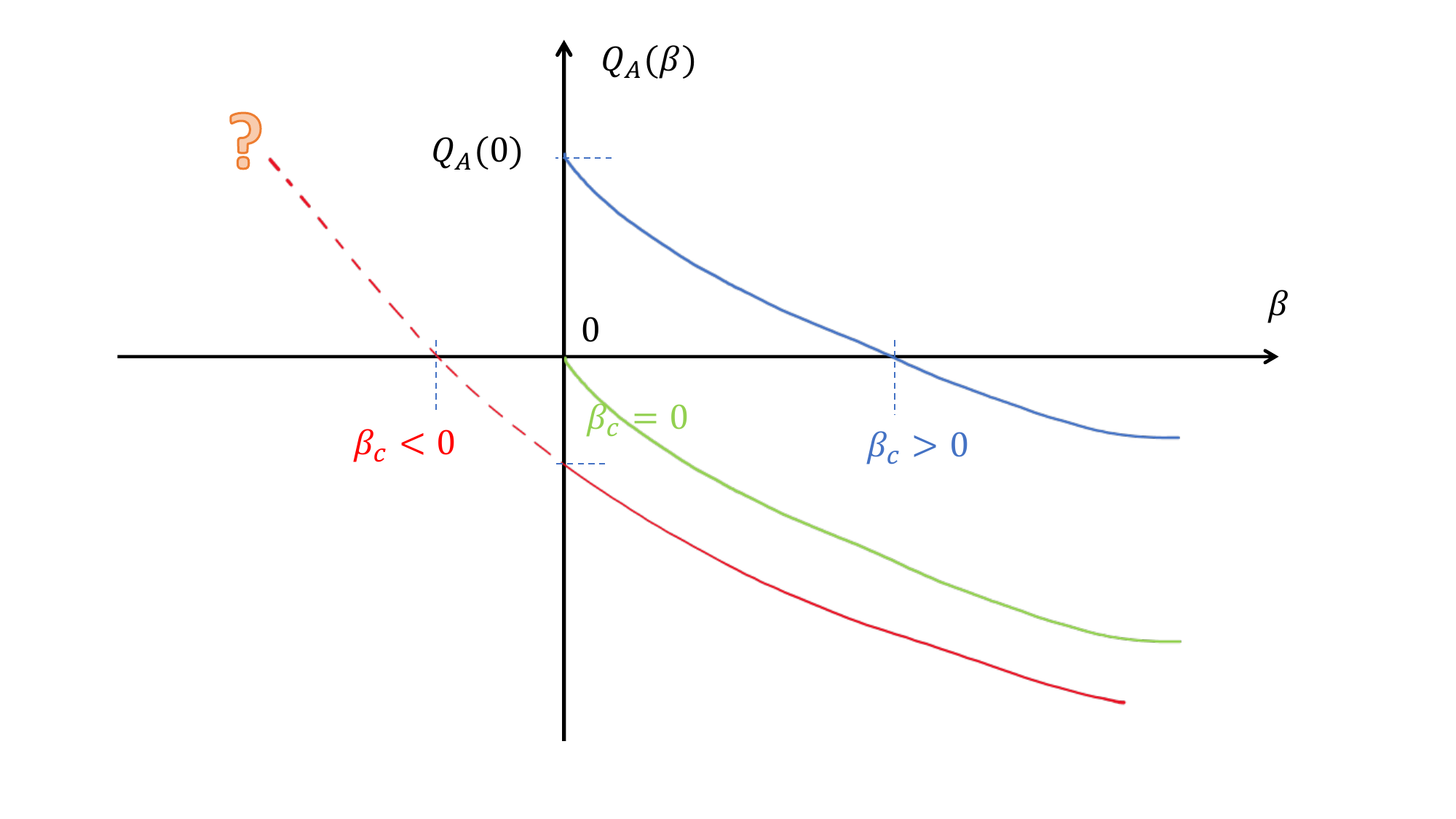}
\caption{The energy flow $Q_\a(\beta)$ (see also the definition \eqref{EvarA}), is a strictly decreasing, convex function of $\beta>0$ and has a well defined limit value $Q_\a(0)$ as $\beta\rightarrow 0_+$. If $Q_\a(0)>0$ then there is a unique root $\bop>0$ (blue curve). If $Q_\a(0)=0$ then the unique root is $\bop=0$ (green curve). If $Q_\a(0)<0$ then there is no root in $[0,\infty)$ (red curve). However, by extending the function $Q_\a(\beta)$ to a function $\bar Q(\beta)$ also defined for $\beta<0$, we have a unique root $\bop<0$ of the extended function $\bar Q(\beta)$. This is the (negative) contact temperature associated to the system state. The question mark indicates that one has to decide what the extension should be. We argue it should be the flow resulting from the coupling to a thermometer having finitely many levels, which is in the Gibbs equilibrium state with $\beta<0$, when the number of levels is taken to infinity.}
   \label{Fig.1}
\end{figure}

The physically correct approach is to consider first a finite-level thermometer for which the Gibbs state is defined for all $\beta\in\mathbb R$, also for $\beta<0$, and then take the limit of infinitely many levels.  The  {\it finite-level} thermometer is defined by the Hamiltonian $H_{\b,K}$ acting on $\mathbb C^{K+1}$ with eigenvalues $\eta_k=k/N$, $k=0,1,\ldots,K$. The state $\tau_{\beta,K}=e^{-\beta H_{\b,K}}/{\tr}e^{-\beta H_{\b,K}}$ is a well defined density matrix for all $\beta\in\mathbb R$ and determines the heat flow
$$
Q_{\a,K}(\beta) = \tr{\big(U(\rho\otimes\tau_{\beta,K}) U^* H_\a\big)} - \tr{(\rho H_\a)},
$$
where the trace is over $\mathbb C^d\otimes \mathbb C^{K+1}$ (system plus thermometer Hilbert space) and $U$ depends on $K$, as it implements equiprobable energy transitions of $H_\a+H_{\b,K}$.  We show in Theorem \ref{thm:QbetainR} that the $\infty$-level limit
$$
\bar Q_\a(\beta) =\lim_{K\rightarrow\infty} Q_{\a,K}(\beta)
$$
exists and is given by 
$$
\bar Q_\a(\beta) = 
\left\{
\begin{array}{ll}
Q^+_\a(\beta):=\displaystyle\sum_{1\leq j < k \leq d}\frac{e^{-\beta (\ve_k -\ve_j)}p_j-p_k}{k(k-1)} \sum_{l=1}^{k-1}(\ve_k-\ve_l), & \quad \beta>0,\\ [15pt]
Q^-_\a(\beta) := \displaystyle\sum_{1\le k< j\le d} \frac{e^{-\beta(\ve_k-\ve_j)}p_j - p_k}{(d-k+1)(d-k)} \sum_{l=k+1}^d(\ve_k-\ve_l), & \quad \beta\leq 0.
\end{array}
\right.
$$
This is a continuous extension of $Q_\a(\beta)$, originally defined for $\beta>0$, now extended to $\beta\in\mathbb R$. The extended function is strictly decreasing (unless $\rho$ is the ground or top excited state of $H_\a$) and satisfies $|\bar Q_\a(\beta)|\leq \ve_d-\ve_1$ for all $\beta\in \mathbb R$. It has a unique root $\bop\in\mathbb R$ which coincides with the previous non-negative root of $Q_\a$ in case $Q_\a(0)\ge 0$, and it also yields a unique negative root which is defined to be the contact temperature of states $\rho$ for which $Q_\a(0)<0$. 
\medskip

{\bf Structural results.} The heat flow $Q_\a(\beta)$ is determined by the reduction to the system part alone, of a thermal process ($U$) acting on the system+thermometer complex. We study this structure in a more general setting in Sections  \ref{sec:eflow} and \ref{secRedQChannels}. In the latter section, we consider a bipartite system with Hilbert space $\h\otimes\h_\rr$  (`system plus reservoir') with  $\dim\h=d<\infty$, $\dim\h_\rr\le \infty$, having a non-interacting Hamiltonian $H+H_\rr$. We take a unitary map $V$ on $\h\otimes\h_\rr$ such that $[H+H_\rr, V]=0$ (thermal process). We define the completely positive and trace preserving (CPTP) map acting on operators $X\in\mathcal B(\h)$ on the system, by 
$$
\Phi_\beta(X) = {\rm Tr}_\rr\big(V(X\otimes\tau_\beta)V^*\big),
$$
where $\tau_\beta\propto e^{-\beta H_\rr}$ is the Gibbs state of $\rr$ (and $\beta>0$). Suppose that the eigenvalues $\ve_j$ of $H$ are simple. Due to the conservation of total energy by $V$, the  thermal channel $\Phi_\beta$ leaves the subspace $\mathcal D$ of diagonal operators (in the eigenbasis of $H$) invariant, and it leaves the off-diagonal operators $\mathcal O$ invariant as well (see Lemma \ref{phiprop}). $\Phi_\beta$ has thus the block decomposition 
$$
\Phi_\beta = \Phi_\beta|_{\mathcal D}\oplus\Phi_\beta|_{\mathcal O}.
$$
Denote by $\psi_j$ the eigenbasis of the system Hamiltonian $H$. The probability of transition $\psi_{j'}\rightarrow \psi_j$ in the system, due to the contact with $\rr$ (mediated by $V$), is given by
$$
G_{jj'} (\beta)= {\rm Tr}\Big(|\psi_j\rangle\langle\psi_j| \Phi_\beta\big(|\psi_{j'}\rangle\langle \psi_{j'}|\big)\Big).
$$
This defines a $d\times d$ matrix  $G(\beta)=(G_{jj'}(\beta))$ with $G_{jj'}(\beta)\ge 0$ and $\sum_{j=1}^d G_{jj'}(\beta)=1$ for all $j'$, which represents the restriction $\Phi_\beta|_{\mathcal D}$. So $G(\beta)$ is a stochastic matrix, called the stochastic matrix associated to the quantum channel $\Phi_\beta$. The vector $v_\beta\in{\mathbb R}^d$ with components $\propto e^{-\beta\ve_j}$, $j=1,\ldots,d$ represents the system equilibrium (Gibbs) density matrix $\propto e^{-\beta H}$. It is a fixed point of the stochastic matrix, $G(\beta)v_\beta=v_\beta$. It follows from a Perron-Frobenius argument (see Corollary \ref{cor1}) that if $G_{jj'}(\beta)>0$ for all $j,j'$ and if all the Bohr energies $\ve_j-\ve_k$ of $H$ are distinct ($j\neq k$), then the eigenvalue $1$ of $\Phi_\beta$ is simple (eigenvector $v_\beta$) and all other eigenvalues have modulus strictly less than $1$. These general properties hold in particular for the specific model of the thermometer coupled to the system.

We show in Proposition \ref{gjjp} that for the system-thermometer model, the heat flow and the stochastic matrix are related for  $\beta>0$ as,
$$
Q_\a(\beta) = \vve \cdot G(\beta) \vvp -\vve\cdot\vvp, 
$$
where $\vve=(\ve_1,\ldots,\ve_{d})^T$ and $\vvp=(p_1,\ldots,p_d)^T$ are the column vectors in $\mathbb R^{d}$ whose entries are the energies of $H_\a$ and populations of $\rho$, and where $G(\beta)$ is the $d\times d$ stochastic matrix associated to $\Phi_\beta$. Moreover, the matrix elements of $G(\beta)$ are given by
$$
G_{j j'}(\beta) =\sum_{s=\max\{j,j'\}}^{d}\frac{1}{s}\big(e^{-\beta (\ve_s-\ve_{j'})}-e^{-\beta (\ve_{s+1}-\ve_{j'})}\big)\ \ \  \text{with \ }\ e^{-\beta \ve_{d+1}}=0,
$$
from which we see directly that $G_{jj'}(\beta)>0$, and the eigenvalues of $G(\beta)$ are strictly positive. The general theory of reduced thermal channels, as outlined above, then implies the following result (Theorem \ref{applithm}):  If all nonzero Bohr energies $\ve_j-\ve_k$ of $H_\a$ are distinct, then $\Phi_\beta|_{\mathcal D}$ has a simple eigenvalue $1$ with associated spectral projection $|v_\beta\rangle\langle v_0|$  ($v_0$ is the vector in $\mathbb C^d$ with elements all equal to $1$) and all other eigenvalues of $\Phi_\beta|_{\mathcal D}$, as well as all eigenvalues of $\Phi_\beta|_{\mathcal O}$, have modulus strictly less than $1$.  The said spectral property implies that the repeated contact of the system with fresh thermometers at temperature $\beta$, resulting in the repeated application of $\Phi_\beta$ to any initial state $\rho$, will drive the system to the Gibbs state $\gamma_\beta=e^{-\beta H}/{\rm Tr}e^{-\beta H}$ at the temperature of the thermometers. This in turn means that the contact temperature $\bop^{(k)}$ of $\rho^{(k)}=\Phi_\beta^k(\rho)$ satisfies $\lim_{k\rightarrow\infty}\bop^{(k)}=\beta$. We discuss the details of this dynamical process of approach to equilibrium in Section \ref{sec:appreq}.

\section{Energy flow between system and thermometer}
\label{sec:model}

Let $\a$ be a quantum system defined on the Hilbert space $\h_\a=\mathbb C^{d}$ with Hamiltonian 
\begin{equation}
H_\a = \sum_{j=1}^{d} \ve_j |\psi_j^\a\rangle\langle\psi_j^\a|,
\end{equation}
where the $\ve_j$ and $\psi_j^\a$ are the energies and energy eigenstates. We assume
\begin{itemize}
\item[{\bf (ND)}]
The energies are ordered, $\ve_1 < \ve_2<\cdots<\ve_{d}
$ and non-degenerate. 
\end{itemize}
States of the system $\a$ are density matrices $\rho$ on the system Hilbert space $\h_\a$ (non-negative operators of unit trace). 
The equilibrium state at temperature $T=1/\beta>0$ is given by the Gibbs density matrix
\begin{equation}
\label{AGibbs}
\gamma_{\beta}=\frac{e^{-\beta H_\a}}{Z_{\a}},
\end{equation}
where $Z_\a={\rm Tr}e^{-\beta H_\a}$ is the partition function. We call this $\beta$ (or $T=1/\beta$) the {\it equilibrium temperature}.

The Hamiltonian of the bath $\b$, also called the thermometer, is given by
\begin{equation}
\label{7}
H_\b = \sum_{k=0}^{\infty} \sum_{m=1}^{m_\b}\eta_k |\psi^\b_{k, m}\rangle\langle\psi^\b_{k,m}|=  \sum_{k=0}^{\infty} \eta_k P_k^\b,
\end{equation}
where the $\eta_k\ge0$ are the eigenvalues, all having degeneracy $1\leq m_\b <\infty$ and where $P_k^\b$ is the (rank $m_\b$) spectral projection associated to $\eta_k$. The vectors $\{\psi_{k,m}^\b\}$ form an orthonormal basis of the (infinite-dimensional) bath Hilbert space $\h_\b$, and the domain of $H_\b$ is the subspace  ${\rm dom}(H_\b):=\{\Psi\in \h_\b \, :\, \sum_{k=0}^{\infty} \sum_{m=1}^{m_\b}|\langle\psi_{k,m}^\b|\Psi\rangle|^2\eta_k^2<\infty\}$. We take the spectrum of the bath to be equally spaced, with small gaps,
\begin{itemize}
\item[{\bf (S)}]
$\eta_k=\frac{k}{N}$, $k\in \mathbb N\cup\{0\}$ for an integer $N$. Moreover, the system energies satisfy $\ve_s = \ve_1+\frac{l_s}{N}$ for some integers $l_s$, $s=1,\ldots,d$, with $0=l_1<l_2<\cdots<l_d$.
\end{itemize}

As $N(\ve_s-\ve_1)=l_s\ge s-1$ we must have $N\ge\max_{2\le s\le d}\frac{s-1}{\ve_s-\ve_1}$. In fact, we are interested in taking $N$ very large ($\rightarrow\infty$), so that $\rm B$ has the characteristics of a macroscopic (reservoir) system. The part of condition (S) saying that the Bohr energies $\ve_s-\ve_{s'}$ of $\a$ are multiples of the basic Bohr energy $1/N$ of the thermometer $\b$ guarantees that $\a$ and $\b$ are able to exchange energy which causes an energy flow. The thermal equilibrium state of $\b$ is given by
\begin{equation}
\label{thermoB}
\tau_{\beta} =\frac{e^{-\beta H_\b}}{Z_{\b}}, \ \ \beta>0,
\end{equation}
where $Z_{\b}=\tr (e^{-\beta H_\b})$ and $\beta$ is the reservoir equilibrium temperature.  The joint system $\a+\b$ is described by the Hilbert space $\h_{\a\b}=\h_\a\otimes\h_\b$. Its uncoupled Hamiltonian (dropping $\otimes \bbbone$ and $\bbbone\otimes$ from the notation) is
$$
H_{\a\b}= H_\a+H_\b,
$$
with domain $\h_\a\otimes {\rm dom}(H_\b)$. 
The interaction between $\a$ and $\b$ is given by a thermal process (thermal quantum operation) $U$, which is a unitary operator on $\h_{\a\b}$ satisfying
\begin{equation}
\label{m1}
[U, H_{\a\b}]=0.
\end{equation}
The energy flow $Q_\a(\beta)$ is defined as the difference of the energy of $\a$ after and before the interaction with $\b$,
\begin{equation}\label{EvarA}
Q_\a(\beta) = \tr{\big(U(\rho\otimes\tau_{\beta}) U^* H_\a\big)} - \tr{(\rho H_\a)}.
\end{equation}
The first trace is over $\h_{\a\b}$, the second one over $\h_\a$. The flow \eqref{EvarA} is invariant under translations of $H_\a$ by $\alpha \mathbb I$ for $\alpha\in \mathbb R$. Due to \eqref{m1} and the invariance of the trace under unitary conjugation, we have $\tr(U(\rho\otimes\tau_\beta)U^*(H_\a+H_\b))= \tr(\rho H_\a)+\tr(\tau_\beta H_\b)$. This shows that $U(\rho\otimes\tau_\beta)U^*H_\b$ has finite trace and that we have $Q_A(\beta)+Q_B(\beta)=0$, where $Q_B(\beta) = \tr{(U(\rho\otimes\tau_{\beta}) U^* H_\b)} - \tr{(\tau_{\beta} H_\b)}$. If $Q_\a(\beta)>0$ then $\a$ receives energy due to the interaction with $\b$. In the opposite case $\a$ gives energy to $\b$.

We introduce a genericity assumption on the thermal process $U$ which can be formulated in two equivalent ways: either in terms of uniformity or in terms of randomness. We show in Theorem \ref{Qunbounded} that they are equivalent, for the purpose of the energy flow. The genericity assumption in terms of uniformity is:
\begin{itemize}
\item[{\bf (G)}]
The transition probability induced by $U$, between any two equal energy eigenstates of the uncoupled Hamiltonian $H_{\a\b}$, are the same.
\end{itemize}
Denoting the distinct eigenvalues of $H_{\a\b} = H_\a+H_\b$ by $E_\ell$, $\ell\in \mathbb N$ we have the spectral representation
\begin{equation}
\label{m3'}
H_{\a\b} = \sum_{\ell=0}^\infty E_\ell P_\ell.
\end{equation}
The $P_\ell$ are the eigenprojections with $\dim P_\ell=g_\ell\ge 1$, where $g_\ell$ is the (finite) degeneracy of the eigenvalues $E_\ell$. The Hilbert space is partitioned into spectral subspaces 
$$
\h_{\a\b} = \bigoplus_{\ell=0}^\infty \h_{\a\b}^{(\ell)},\qquad \h_{\a\b}^{(\ell)} = {\rm Ran}P_\ell.
$$ 
Due to \eqref{m1} we have 
$U=\bigoplus_{\ell =0}^\infty U^{(\ell)}$ for unitaries $U^{(\ell)}$ on $\h_{\a\b}^{(\ell)}$.
The genericity assumption in terms of randomness is the following:
\begin{itemize}
\item[\bf (R)] The $U^{(\ell)}$ are independent random unitaries for different $\ell$. For any fixed $\ell$,  $U^{(\ell)}$ is distributed according to the Haar measure on the group of unitaries acting on $\h_{\a\b}^{(\ell)}$. 
\end{itemize}
The Haar measure is the unique probability measure on the unitary group that is invariant under left and right multiplication by any unitary matrix. 
\medskip
 
 Under the assumption {\bf (R)} the energy flow \eqref{EvarA}  is a random quantity $Q_A(\beta)=Q_\a(\beta,\omega)$, where $\omega$ indicates the randomness parameter. The next result gives an explicit form of the flow and shows that the average (expectation), $\mathbb E[Q_\a(\beta,\omega)]$ is {\it the same} as the $Q_\a(\beta)$, \eqref{EvarA} obtained under the genericity assumption {\bf (G)}.

\begin{thm}[Explicit form of the energy flow]
\label{Qunbounded}
Assume that either {\bf (G)} or {\bf (R)} hold. Let $Q_\a(\beta)$ denote the flow \eqref{EvarA} in the former case, and let $Q_\a(\beta)$ denote the expectation of the flow \eqref{EvarA} in the latter case. Then we have for any $\beta>0$,
\begin{align}\label{qabeta}
Q_\a(\beta)&=\sum_{1\leq j < k \leq d}\frac{e^{-\beta (\ve_k -\ve_j)}}{k(k-1)}p_j \sum_{l=1}^{k-1}(\ve_k-\ve_l)- \sum_{k=2}^{d}\frac{p_k}{k}\sum_{l=1}^{k-1}(\ve_k-\ve_l).
 \end{align}
Moreover, 
\begin{equation}
\label{qa0}
Q_\a(0) :=\lim_{\beta\rightarrow 0_+} Q_\a(\beta) =  \sum_{j=1}^d \ve_j(\frac1d -p_j).
\end{equation}
\end{thm}

We prove Theorem \ref{Qunbounded} in Section \ref{sec:proothm1}.
\medskip

{\bf Discussion and remarks.} (i) It is manifest from \eqref{qabeta} that the flow $Q_\a(\beta)$ is independent of the spectral gap $1/N$ of $\b$ as well as the bath multiplicity $m_\b$ (see \eqref{7}). It can be written equivalently as 
\begin{equation}
\label{qabeta'}
Q_\a(\beta)=\sum_{1\leq j < k \leq d}\frac{e^{-\beta (\ve_k -\ve_j)}p_j-p_k}{k(k-1)} \sum_{l=1}^{k-1}(\ve_k-\ve_l).
\end{equation}

(ii) As the $\ve_k$ are ordered, \eqref{qabeta} is the difference of two non-negative terms, one depending on $\beta$, the other not. The sign of $Q_\a(\beta)$ is thus determined by the relative size of these two terms, and $Q_\a(\beta)=0$ when the two are equal. 

(iii) The first term in \eqref{qabeta}  does not depend on $p_d$, the population of the most excited state. Is is strictly positive unless $p_d=1$, in which case it vanishes. The second term in \eqref{qabeta} which is subtracted from the first one, does not depend on the ground state population $p_1$. It is strictly positive unless $p_1=1$, in which case it vanishes.

(iv) The expression \eqref{qabeta'} shows that $Q_A(\beta)>0$ if $e^{\beta\ve_k}p_k$ is decreasing in $k$ and $Q_A(\beta)<0$ if $e^{\beta\ve_k}p_k$ is increasing in $k$. For $e^{\beta\ve_k} p_k$ constant in $k$, {\it i.e.}~$p_k=e^{-\beta \ve_k}/Z_\a$ (Gibbs state) we have $Q_\a(\beta)=0$.

\section{The contact temperature}
\label{sec:contt}

Suppose the flow $Q_\a(\beta)$, given by \eqref{EvarA} for $\beta>0$, has a unique root $\bop>0$. Then this root is the contact temperature of the state $\rho$ of $\a$.  The function $\beta\mapsto Q_\a(\beta)$ has the limit \eqref{qa0} as $\beta\rightarrow 0_+$. Physically, $Q_\a(0)$ is the energy flow of the system $\a$ (positive for an inflow of energy) caused by the contact with a thermometer $\b$ which has infinite equilibrium temperature $T=1/\beta=\infty$. The second law of thermodynamics says that energy will flow from a hotter body to a colder body | if both bodies are in (or close to) thermal equilibrium. In this situation we must have $Q_\a(0)>0$. However, if the state $\rho$ of $\a$ is {\it not} close to equilibrium, then there is no reason to expect $Q_\a(0)>0$ and indeed, \eqref{qa0} characterizes all $\rho$ with $Q_\a(0)<0$. Such $\rho$ are  `far from equilibrium' and they violate the second law of equilibrium thermodynamics  | the system $\a$ in such a state will give energy when coupled to a thermometer in equilibrium at arbitrarily high temperature. This is counterintuitive only if we follow the erroneous instinct to apply the second law of thermodynamics to systems far from equilibrium.
\smallskip

{\bf Two extreme cases:~ground and most excited states.}  The energy flow associated to the ground state $\rho=|\psi^\a_1\rangle\langle\psi^\a_1|$ is 
$$
Q_\a(\beta) = \sum_{k=2}^d \frac{e^{-\beta (\ve_k-\ve_1)}}{k(k-1)} \sum_{l=1}^{k-1}(\ve_k-\ve_l)>0, 
$$
which is strictly positive for all $\beta>0$ and satisfies  $\lim_{\beta\rightarrow\infty}Q_\a(\beta)=0$. So it makes sense to associate $\bop=\infty$ to the ground state of $\a$. For the most excited state $\rho=|\psi^\a_d\rangle\langle\psi^\a_d|$, \eqref{qabeta} gives 
\begin{equation}
\label{extst}
Q_\a(\beta) = -\frac1d\sum_{l=1}^{d-1}(\ve_d-\ve_l)<0. 
\end{equation}
This is a strictly negative constant, independent of $\beta$ and it has no root. We will see in Theorem \ref{thmQbar} that it is natural to assign $\bop=-\infty$ to the most excited state.

\medskip

We show in the next two sections that there is a finite $\bop\in\mathbb R$ for {\it any} state $\rho$ other than the ground and top excited states. Moreover, $\bop>0$ for $Q_\a(0)>0$ and $\bop<0$ for $Q_\a(0)<0$, while $Q_\a(0)=0$ corresponds to $\bop=0$.

\subsection{Positive contact temperature}

As discussed above, $Q_\a(0)\ge 0$ is a consequence of the second law of thermodynamics {\rm if} the state $\rho$ is close to equilibrium. Here we assign a $\bop\ge 0$ to any state $\rho$ with $Q_\a(0)\ge 0$.

\begin{thm}[Existence and regularity of non-negative contact temperature]
\label{thm2} Let $\rho$ be any state of $\a$ other than the ground state or the most excited state of $H_\a$.
\begin{itemize}
\item[{\rm (i)}] If $Q_\a(0)>0$ then there exists a unique $\bop>0$ such that $Q_\a(\bop)=0$.  

Moreover, $\bop$ is a real analytic function of the populations of $\rho$, that is of the probability vector $\{p_j\}_{j=1}^d$, on the open simplex ${\mathcal S}=\{ (p_1, p_2, \dots, p_d)\in (\mathbb R_+\backslash\{0\})^d\ | \sum_{j=1}^dp_j=1 \}$.

\item[{\rm (ii)}] If $Q_\a(0)=0$ then $\bop=0$ is the unique root of $Q_\a(\beta)$ in $[0,\infty)$.
\end{itemize}
\end{thm}

{\bf Passive and active states.} Consider a Hamiltonian $H={\rm diag}(E_1,\ldots,E_d)$, with possibly degenerate energies. A state $\rho$ is called passive if and only if $\tr(\rho H)\le \tr(V\rho\, V^* H)$ for all unitaries $V$. If $\rho$ is not a passive state, then it is called an active state. One can show that $\rho$ is passive if and only if $\rho={\rm diag}(p_1,\ldots,p_d)$ is diagonal in the $H$ basis and $E_j\le E_k\Rightarrow p_j\ge p_k$ \cite{KouKou}. For our model, it follows from \eqref{qa0} and the fact that the eigenvalues $\ve_j$ of $H_\a$ are ordered (see the condition (ND)), that $Q_\a(0)\ge 0$  for all passive states, see the Remark (iv) below Theorem \ref{Qunbounded} for $\beta=0$.
\begin{cor}
\label{pass/act:corollary}
Let $\rho$ be a passive state other than the ground state of $H_\a$. Then  $\bop> 0$. There are active states for which $\bop>0$. 
\end{cor}
{\bf Example.} One may rewrite \eqref{qa0} as
$$
Q_\a(0) = \sum_{j=1}^{d-1} \Big(p_j-\frac 1d\Big)(\ve_d-\ve_j).
$$
For a qubit ($d=2$) we have $Q_\a(0)\ge 0 \Leftrightarrow p_1\ge p_2$, with equivalence of the equalities. So $\bop\ge 0$ is equivalent with a diagonal $\rho$ being a passive state. For $d=3$, $Q_\a(0)=(p_1-\frac13)(\ve_3-\ve_1)+(p_2-\frac13)(\ve_3-\ve_2)$. By choosing $p_2>p_1$ and $\ve_1<\ve_2<\ve_3$ such that $Q_\a(0)>0$ we obtain an active state $\rho={\rm diag}(p_1,p_2,1-p_1-p_2)$ for which $\bop>0$. In particular, any $p_1,p_2$ with $\frac13<p_1<p_2$ will give such an active state.

\medskip

{\bf Proof of Theorem \ref{thm2}.} The proof of Theorem \ref{thm2} is simple. Viewing $Q_\a(\beta)$, \eqref{qabeta} as a function of $\beta\in\mathbb R$, we have for $p_1\neq 1$, $p_d\neq 1$,
$$
\lim_{\beta\rightarrow -\infty} Q_\a(\beta) = \infty,\qquad \lim_{\beta\rightarrow \infty} Q_\a(\beta) = - \sum_{k=2}^{d}\frac{p_k}{k}\sum_{\ell=1}^{k-1}(\ve_k-\ve_\ell)<0.
$$
Consequently as the function $\beta\mapsto Q_\a(\beta)$ is continuous it must have at least one root in $\mathbb R$. This root is actually unique, because the function is strictly decreasing. In fact, it follows directly from \eqref{qabeta} that for $n=1,2,\ldots$
(and $p_1\neq 1$, $p_d\neq 1$),
\begin{align}
(-1)^n\partial_\beta^n Q_\a(\beta)>0.
\end{align}
We see directly that the unique root satisfies $\bop>0\Leftrightarrow Q_\a(0)>0$ | see also Fig.~\!\ref{Fig.1}. 

Let us turn to the regularity of $\bop$ in the probability vector $\mathbf p:=\{p_j\}_{j=1}^d$ for $Q_\a(0)>0$. We write $Q_\a(\mathbf p,\beta)$ when we consider $\mathbf p$ as a variable. The expression \eqref{qabeta} gives the extension of $Q_\a(\mathbf p,\beta)$ to $]0,1[^d\times \mathbb R^*_+$ and being linear in $\mathbf p$, this function admits an analytic extension to $\mathbb C^d\times \mathbb C$. By the argument above, for any $\mathbf p\in ]0,1[^d$, there exists a unique $\beta(\mathbf p)>0$ such that $Q_\a(\mathbf p,\beta(\mathbf p))=0$. The analytic implicit function theorem (see {\it e.g.} \cite{D}, Chap. X, \S 2, (10.2.1) and (10.2.4)) implies that the map $]0,1[^d\ni \mathbf p\mapsto \beta(\mathbf p)$ is real analytic. Restricting this function $\beta$ to $\mathcal S\subset ]0,1[^d$ by expressing one $p_k$ as a function of the others, we get a real analytic function which coincides with $\bop(\mathbf p)$.
\hfill \qed
\medskip

{\it Remark:} As seen in point (iv) after Theorem \ref{Qunbounded} we have $Q_\a(\mathbf p_{\beta_0},\beta_0)=0$ for any $\beta_0>0$, where we use the notation of the proof of Theorem \ref{thm2}, with $(\mathbf p_\beta)_k=e^{-\beta \ve_k}/Z_\a$, $k=1,\dots, d$ representing the Gibbs thermal distribution. As $\mathbf p\mapsto Q_\a(\mathbf p,\beta_0)$ is linear in $\mathbf p$, the set of all probability distributions $\mathbf p\in\bar{\mathcal S}$ associated to the same contact temperature $\beta_0$, $\bop^{-1}(\{\beta_0\})=\{\mathbf p_{\beta_0}+ \ker Q_\a(\cdot , \beta_0)\}\cap \bar {\mathcal{S}}$, is a section of the simplex $\bar{\mathcal S}$ of dimension $d-2$.

\subsection{\bf Negative contact temperature}

We now consider the situation $Q_\a(0)<0$ which is complementary to that treated in Theorem \ref{thm2}. 
To derive the existence of $\bop\ge 0$ in Theorem \ref{thm2} we argued that we may consider $Q_\a(\beta)$ as a function of $\beta\in\mathbb R$ | simply extending the expression \eqref{qabeta} of Theorem \ref{Qunbounded} to all $\beta\in\mathbb R$. For $\beta>0$ the quantity $Q_\a(\beta)$ has a physical meaning (difference of energy of $\a$ before and after the effect of the thermometer). But for $\beta<0$ it does not. Even though the expression \eqref{qabeta} is well defined (even analytic) for all $\beta$, as discussed above, it is unphysical for $\beta<0$.

\smallskip

Consider then a truncated version of the thermometer Hamiltonian given in \eqref{7},
\begin{equation}
\label{f1}
H_{\b,K} = \sum_{k=0}^K \eta_k |\psi^\b_k\rangle\langle\psi_k^\b|,\qquad \eta_k = \frac kN,
\end{equation}
where $K\in\mathbb N\cup\{0\}$, and which acts on the truncated Hilbert space $\mathbb C^{K+1}$. We focus on the nondegenerate case ($m_\b=1$). The thermometer Gibbs state can now be defined for any $\beta\in\mathbb R$ as in \eqref{thermoB},
\begin{equation}
\label{f15}
\tau_{\beta,K}= \frac{e^{-\beta H_{\b,K}}}{Z_{\b,K}},\qquad \text{with}\qquad Z_{\b,K} = \frac{1-e^{-\frac{\beta}{N}(K+1)}}{1-e^{-\frac{\beta}{N}}}.
\end{equation}
(For $\beta=0$ the formula for $Z_{\b,K}$ reduces to $Z_{\b,K}=K+1$.)  The energy flow of $\a$ is defined as in \eqref{EvarA} with $\tau_\beta$ replaced by $\tau_{\beta,K}$,
\begin{equation}
\label{EvarAK}
Q_{\a,K} = \tr{\big(U(\rho\otimes\tau_{\beta,K}) U^* H_\a\big)} - \tr{(\rho H_\a)},
\end{equation}
where $U$ satisfies \eqref{m1} with $H_{\a\b}=H_\a+H_{\b,K}$, and thus depends on $K$. We define the flow in the limit of infinitely many levels of the thermometer,
\begin{equation}
\bar Q_\a(\beta) =\lim_{K\rightarrow\infty}Q_{\a,K}(\beta).
\end{equation}
Then we have the following result 
\begin{thm}
\label{thmQbar}
\ 
\begin{itemize}
\item[\rm \bf (1)] Suppose $\rho$ is neither the ground nor the most excited state of $H_\a$. Then there exists a unique $\bop\in\mathbb R$ such that $\bar Q_\a(\bop)=0$. Moreover, $\bop$ has the same sign as $\bar Q_\a(0)$ and $\bop=0$ if and only if $\bar Q_\a(0)=0$. 

For $\bar Q_\a(0)\neq 0$, $\bop$ is a real analytic function of the probability vector $\mathbf p=\{p_k\}_{k=1}^d$ (the popluations of $\rho$) on the open simplex  ${\mathcal S}=\{ (p_1, p_2, \dots, p_d)\in (\mathbb R_+\backslash\{0\})^d\ | \sum_{j=1}^dp_j=1 \}$.

\item[\rm\bf (2)] If $\rho$ is the ground state, then $\bar Q_\a(\beta)>0$ for all $\beta\in \mathbb R,$ and $\lim_{\beta\rightarrow\infty}\bar Q_\a(\beta)=0$. If $\rho$ is the most excited state, then $\bar Q_\a(\beta)<0$ for all $\beta\in \mathbb R,$ and $\lim_{\beta\rightarrow-\infty}\bar Q_\a(\beta)=0$.
\end{itemize}
\end{thm}

In accordance with statement (2) of Theorem \ref{thmQbar} we assign $\bop=\infty$ to the ground state and $\bop=-\infty$ to the most excited state.
\medskip

Theorem \ref{thmQbar} is a consequence of the following result, proven in Section \ref{sec:proofthmQbar}. 

\begin{thm}
\label{thm:QbetainR}
Let $0\neq\beta\in\mathbb R$. Then we have 
\begin{equation}
\label{f34}
\bar Q_\a(\beta) = 
\left\{
\begin{array}{ll}
Q^+_\a(\beta):=\displaystyle\sum_{1\leq j < k \leq d}\frac{e^{-\beta (\ve_k -\ve_j)}p_j-p_k}{k(k-1)} \sum_{l=1}^{k-1}(\ve_k-\ve_l), & \quad \beta>0,\\ [15pt]
Q^-_\a(\beta) := \displaystyle\sum_{1\le k< j\le d} \frac{e^{-\beta(\ve_k-\ve_j)}p_j - p_k}{(d-k+1)(d-k)} \sum_{l=k+1}^d(\ve_k-\ve_l), & \quad \beta<0.
\end{array}
\right.
\end{equation}
Both $Q^\pm_\a(\beta)$ are defined for $\beta\in\mathbb R$ (are in fact analytic entire in $\beta$) and we have 
\begin{equation}
\label{contflowzero'}
Q^+_\a(0)=Q^-_\a(0) = \sum_{j=1}^d \ve_j(\frac1d -p_j),
\end{equation}
which also coincides with the limit of $Q_{\a,K}(0)$ as $K\rightarrow\infty$.
\end{thm}

The $\bar Q_\a(\beta)$, $\beta\in\mathbb R$ given in \eqref{f34} is a natural extension of $Q_\a(\beta)$, $\beta>0$ of \eqref{qabeta}. To derive Theorem \ref{thmQbar} we note that $\bar Q_\a(\beta)$ is continuous for $\beta\in\mathbb R$, even at $\beta=0$. Moreover, the explicit expressions show that both $Q^-_\a(\beta)$ and $Q^+_\a(\beta)$ are strictly decreasing in $\beta\in\mathbb R$ (unless $\rho$ is the ground or the top excited state, in which case $Q_\a^-(\beta)$ or $Q_\a^+(\beta)$ are constant, respectively), and that $Q^-_\a(\beta)$ is concave and $Q^+_\a(\beta)$ is convex.  Moreover,
\begin{align}
\lim_{\beta\rightarrow +\infty}\bar Q_\a(\beta) &= -\sum_{1\leq j < k \leq d}\frac{p_k}{k(k-1)} \sum_{l=1}^{k-1}(\ve_k-\ve_l)\le 0, \label{18.1}\\
\lim_{\beta\rightarrow -\infty} \bar Q_\a(\beta) &= -\sum_{1\le k< j\le d} \frac{p_k}{(d-k+1)(d-k)} \sum_{l=k+1}^d(\ve_k-\ve_l) \ge 0.\label{18.2}
\end{align}
The value of $\bar Q_\a(\beta)$ now stays bounded as $\beta\rightarrow-\infty$, as it is truly an average energy difference of the finite system $\a$. If $\rho$ is not the ground state, then the inequality of \eqref{18.1} is strict (and it is an equality for the ground state). If $\rho$ is not the most excited state, then the inequality \eqref{18.2} is strict (and it is an equality for the most excited state). So for $\rho$ none of these two states, $\bar Q_\a(\beta)$, \eqref{f34} has a unique root $\bop\in\mathbb R$ whose sign is the same as that of $\bar Q_\a(0)$, \eqref{contflowzero'}. The regularity of $\bop\neq 0$ for $\bar Q_\a(0)\neq 0$  as a function of $\mathbf p\in \mathcal{S}$ is proven as in the case of Theorem \ref{thm2}.
\medskip

{\bf Examples.} {\bf (1)} Consider first the qubit case $d=2$, and $\rho$ such that $p_1p_2\neq 0$, to avoid the extreme cases dealt with above.  We compute
$$
Q^+_\a(\beta)= \frac{\varepsilon_2-\varepsilon_1}{2}\big ( e^{-\beta (\varepsilon_2-\varepsilon_1)}p_1-p_2\big),\qquad 
    Q^-_\a(\beta)= \frac{\varepsilon_1-\varepsilon_2}{2}\big ( e^{-\beta (\varepsilon_1-\varepsilon_2)}p_2-p_1\big)
$$
and $Q^\pm_\a(0)= \frac{\varepsilon_2-\varepsilon_1}{2}(p_1-p_2)$. We point out that  $\partial_\beta Q_\a^-(\beta)|_{\beta=0}=-\frac{(\ve_2-\ve_1)^2}{2}p_2$ and $\partial_\beta Q_\a^+(\beta)|_{\beta=0}=-\frac{(\ve_2-\ve_1)^2}{2}p_1$, which shows that $\bar Q_\a$ is not differentiable at $\beta=0$ unless $p_1=p_2$.
In all cases, 
\begin{equation}\label{betaqbit}
\bop=\frac{\ln(p_1)-\ln(p_2)}{\varepsilon_2-\varepsilon_1}.
\end{equation}
Note also that for any fixed $\beta\in \mathbb R$, we have $\bar Q_\a(\beta)=0$ if and only if
\begin{equation}\label{gibbstemp}
    p_j=\frac{e^{-\beta \varepsilon_j}}{e^{-\beta \varepsilon_1}+e^{-\beta \varepsilon_1}}=\braket{\psi_j^\a|\gamma_{\beta}\, \psi_j^\a}, \quad j=1,2.
\end{equation}
In the two-dimensional case, the condition $\bar Q(\beta)=0$ thus fixes the diagonal of $\rho$ in the $H_\a$ basis. This is due to  the fact that for qubits, the trace condition and $\bar Q_\a(\beta)=0$ uniquely determine the state, and that the Gibbs state $\gamma_{\beta}$ always satisfies both constraints. However, there is a non-uniqueness or indeterminacy of the coherence (off-diagonal density matrix element). Indeed, all qubit density matrices with the given diagonal and varying coherence $z$ satisfying $|z|^2\le p_1(1-p_1)$ have the {\it same} associated contact temperature $\bop$. In higher dimensions $\bop$ does not fix the diagonal of the state. 
Finally, we remark that the effective temperature $\beta^*$ defined as the unique zero of
\begin{align}\label{defbeta^*}
    \beta\mapsto \tr (H_\a \rho)-\tr(H_\a \gamma_\beta),
\end{align}
the minimizer of the relative entropy $S(\rho|\gamma_\beta)$, see Section \ref{sec:intro}, coincides in the qubit case with $\bop$ given by \eqref{betaqbit}.

\medskip

{\bf (2)} 
We now turn to $d=3$ to further illustrate properties of $\bop$, showing in particular that $\bop$ and $\beta^*$ are in general unrelated.
Assume $p_1, p_3\neq 1$. In this case, $\bop$ is determined by the vanishing of $\bar Q_\a$ given by \eqref{f34}
\begin{align}
\bar Q_\a^+(\beta)
={}&
\Bigg[
\frac12(\varepsilon_2-\varepsilon_1)e^{-\beta(\varepsilon_2-\varepsilon_1)}
+
\frac16\big(2\varepsilon_3-\varepsilon_1-\varepsilon_2\big)
e^{-\beta(\varepsilon_3-\varepsilon_1)}
\Bigg]p_1
\label{qd=3}\\
&+
\Bigg[
-\frac12(\varepsilon_2-\varepsilon_1)
+
\frac16\big(2\varepsilon_3-\varepsilon_1-\varepsilon_2\big)
e^{-\beta(\varepsilon_3-\varepsilon_2)}
\Bigg]p_2
-
\frac13
\big(2\varepsilon_3-\varepsilon_1-\varepsilon_2\big)
\,p_3\qquad \text{for $\beta>0$}
\nonumber
\end{align}
According to \eqref{defbeta^*}$, \beta^*$ is characterized by 
\begin{equation}
\label{betastar3}
e^{-\beta \varepsilon_1}(\varepsilon_1-\tr(\rho H_\a))+e^{-\beta \varepsilon_2}(\varepsilon_2-\tr(\rho H_\a))+e^{-\beta \varepsilon_3}(\varepsilon_3-\tr(\rho H_\a))=0.
\end{equation}
Consider states satisfying 
\begin{equation}
\label{rest}
    \varepsilon_2=\tr(\rho H_\a)=p_1\varepsilon_1+p_2\varepsilon_2+p_3\varepsilon_3.
\end{equation}
From this and from the trace condition, we deduce that the state can be parameterized by $p_2$,
\begin{align}
\label{paramp2}
0&\leq p_2\leq 1,\ \ 
    p_1=(1-p_2)\frac{\varepsilon_3-\varepsilon_2}{\varepsilon_3-\varepsilon_1}, \ \ 
    p_3=(1-p_2)\frac{\varepsilon_2-\varepsilon_1}{\varepsilon_3-\varepsilon_1}.
\end{align}
\begin{figure}[h!]
\centering
\includegraphics[width=.8\linewidth]{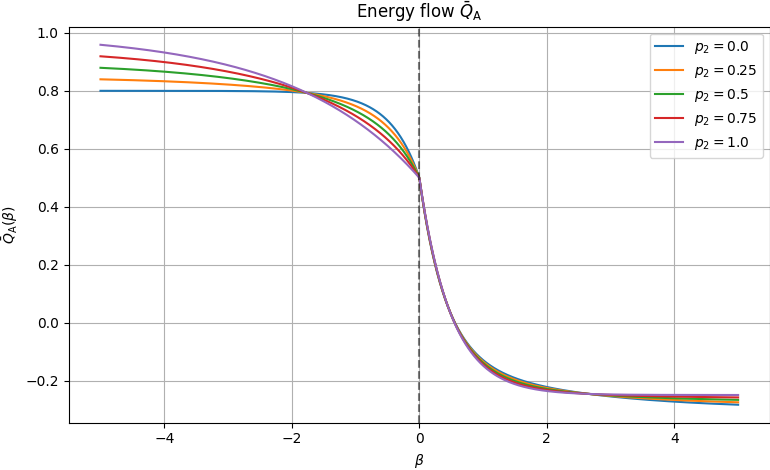}
\caption{The energy flow $\bar Q_\a(\beta)$, \eqref{f34} for states parameterized by $p_2$ as per \eqref{paramp2}, at $\ve_1=0, \ve_2=1/2, \ve_3=5/2$, for various values of $0\leq p_2\leq 1$.}
   \label{Fig.2}
\end{figure}
We deduce from \eqref{betastar3}, \eqref{rest} and \eqref{paramp2} that
\begin{equation}
\label{3levelbetas}
\beta^*=\frac{1}{\varepsilon_3-\varepsilon_1}\ln(\frac{\varepsilon_3-\varepsilon_2}{\varepsilon_2-\varepsilon_1})=\frac{\ln(p_1)-\ln(p_3)}{\varepsilon_3-\varepsilon_1}.
\end{equation}
We now consider the situation for large $\ve_3$. As $\bar Q_\a(0)=(\ve_1+\ve_3-2\ve_2)/3$, which is independent of $p_2$, we have $\bar Q_\a(0)>0$ and $\bop>0$ for all $0\le p_2\le 1$. From \eqref{paramp2}, \eqref{3levelbetas} we have for fixed $0\le p_2\le 1$ and large $\ve_3$,
\begin{align}
\label{Est1}
p_1\sim 1-p_2, & \quad p_3\sim (1-p_2)\frac{\ve_2-\ve_1}{\ve_3},\\
\quad e^{-\beta^*(\ve_3-\ve_1)}\sim \frac{\ve_2-\ve_1}{\ve_3},\quad &e^{-\beta^*(\ve_3-\ve_2)}\sim \frac{\ve_2-\ve_1}{\ve_3},\quad e^{-\beta^*(\ve_2-\ve_1)}\sim 1.
\label{Est1.5}
\end{align}
Using \eqref{Est1}, \eqref{Est1.5} in \eqref{qd=3} gives,
\begin{equation}
\bar Q_\a^+(\beta^*)
\sim \frac16 (\ve_2-\ve_1)(1-2p_2). 
\label{Est2}
\end{equation}
The sign of $\bar Q_\a^+(\beta^*)$ depends on $p_2$:  $\bar Q_\a^+(\beta^*)\ge 0 \Leftrightarrow p_2\le 1/2$ with equality on the left if and only if equality on the right. For $p_2<1/2$ we have $\bar Q_\a^+(\beta^*)>0$ while (by definition) $\bar Q_\a^+(\bop)=0$. As $\mathbb R\ni \beta\mapsto \bar Q_\a^+(\beta)$ is strictly decreasing, this implies that $\beta^*<\bop$. For the same reason we have $\beta^*>\bop$ when $p_2>1/2$. This shows that the two temperatures are not related, generally.

\section{Energy flow, thermal  channel, stochastic matrix}
\label{sec:eflow}

For $\beta>0$, we define the quantum channel $\Phi_\beta: \mathcal B(\h_A)\rightarrow\mathcal B(\h_A)$  by
\begin{equation}
\label{defPhibeta}
\Phi_\beta(X) = {\rm Tr}_B\big(U (X\otimes
\tau_\beta) U^*\big),\qquad X\in\mathcal B(\h_A).
\end{equation}
The CPTP map $\Phi_\beta$ is a particular example of a reduced thermal quantum channel. Due to \eqref{m1} the system Gibbs state \eqref{AGibbs} is invariant under this channel, 
\begin{align}
\label{invargibbs}
    \Phi_\beta(\gamma_\beta)=\gamma_\beta.
\end{align}
We derive general properties of those channels in Section \ref{secRedQChannels}. We associate to $\Phi_\beta$ a $d\times d$ matrix $G$ with matrix elements
\begin{equation}
\label{defGqCh}
G_{jj'} = \big\langle \psi^\a_j \big| \Phi_\beta\big(|\psi^\a_{j'}\rangle\langle\psi^\a_{j'}|\big)\big| \psi^\a_j\big\rangle.
\end{equation}
The matrix elements have the following interpretation. Given that the system $A$ starts in the energy eigenstate $\psi^A_{j'}$ before interacting with the thermometer $B$, the probability of finding the system (upon a quantum measurement of the observable $H_A$) in the eigenstate $\psi^A_j$ after the interaction is just $G_{jj'}$. Therefore we have $G_{jj'}\ge 0$ and $\sum_{j=1}^d G_{jj'}=1$ for all $j'$, which means that $G$ is a stochastic matrix, which we call the stochastic matrix associated to the quantum channel $\Phi_\beta$. The energy flow $Q_A$ is expressed via the stochastic matrix $G$ as follows.

\begin{prop}
\label{gjjp}
With the notation of Theorem \ref{Qunbounded} we have for $\beta>0$
\begin{equation}
\label{flow}
Q_\a = \vve \cdot G \vvp -\vve\cdot\vvp, 
\end{equation}
where $\vve=(\ve_1,\ldots,\ve_{d})^T$ and $\vvp=(p_1,\ldots,p_d)^T$ are the column vectors in $\mathbb R^{d}$ whose entries are the energies of $H_\a$ and populations of $\rho$, and where $G$ is the $d\times d$ stochastic matrix associated to $\Phi_\beta$, \eqref{defGqCh}. Moreover, the matrix elements of $G$ have the following expression,
\begin{equation}
\label{33}
 G_{j j'} =\sum_{s=\max\{j,j'\}}^{d}\frac{1}{s}\big(e^{-\beta (\ve_s-\ve_{j'})}-e^{-\beta (\ve_{s+1}-\ve_{j'})}\big)\ \ \  \text{with \ }\ e^{-\beta \ve_{d+1}}=0,
\end{equation}
and the eigenvalues of $G$ belong to $(0,1]$, $1$ being a simple eigenvalue.
\end{prop}

{\bf Proof of Proposition \ref{gjjp}.} 
Using the orthonormal basis $\psi^\a_j\otimes \psi^\b_{km}$ of $\mathcal H_\a\otimes\mathcal H_\b$ and writing the corresponding matrix elements of $U$ as $U_{jkm,j'k'm'}$, we obtain 
\begin{align}
{\rm Tr}&\big(U(\rho\otimes\tau_{\beta}) U^* H_\a\big) =
\sum_{jkm}\, \ve_j\big\langle\psi_j^\a\otimes\psi^\b_{km}| U(\rho\otimes\tau_\beta)U^*|\psi_j^\a\otimes\psi^\b_{km}\big \rangle\nonumber\\
&=\sum_{jkm}\sum_{j' i'}\sum_{k'm'} \ve_j \langle \psi_{j'}^\a|\rho|\psi^\a_{i'}\rangle \frac{e^{-\beta\eta_{k'}}}{Z_\b} U_{jkm,j'k'm'}\overline{U_{jkm,i'k'm'}}.
\label{15}
\end{align}
Due to \eqref{m1} the product of the matrix elements vanishes unless $\ve_j+\eta_k=\ve_{j'}+\eta_{k'}=\ve_{i'}+\eta_{k'}$. The simplicity of the eigenvalues $\ve_j$  forces $j'=i'$. The expression thus only depends on the diagonal density matrix elements $p_{j'}=\langle\psi_{j'}^\a|\rho|\psi_{j'}^\a\rangle$,
\begin{align}
\label{12}
{\rm Tr}\big(U(\rho\otimes\tau_{\beta}) U^* &H_\a\big)  =\sum_{jj'}\sum_{kk'mm'} \ve_j p_{j'} \frac{e^{-\beta\eta_{k'}}}{Z_\b} |U_{jkm,j'k'm'}|^2\nonumber\\
& = \sum_\ell \sum_{jj'kk'mm'} \chi\big(\ve_j+\eta_k=\ve_{j'}+\eta_{k'}=E_\ell\big)\ve_j p_{j'} \frac{e^{-\beta\eta_{k'}}}{Z_\b} |U^{(\ell)}_{jkm,j'k'm'}|^2.
\end{align}
Next we calculate the modulus square of the density matrix elements. Consider first the assumption {\bf (G)}. For each $\ell$, the transition probabilities between any eigenstates $\psi^\a_j\otimes \psi^\b_{km}$ and $\psi^\a_{j'}\otimes \psi^\b_{k'm'}$ satisfy
\begin{equation}
\label{m2}
\big|U^{(\ell)}_{jkm, j'k'm'}\big|^2\equiv \big|\langle \psi^\a_j\otimes\psi^\b_{km}, U^{(\ell)} \psi^\a_{j'}\otimes \psi^{\b}_{k'm'}\rangle\big|^2 = 1/g_\ell,
\end{equation}
where $g_\ell=\dim\h_{\a\b}^{(\ell)}$ is the size of the matrix. The value of the constant $1/g_\ell$ is determined by the condition that $U^{(\ell)}$ is unitary on $\h_{\a\b}^{(\ell)}$. Indeed, the column $jkm$ must be a vector of norm one, so  $\sum_{j'k'm'} |U^{(\ell)}_{jkm,j'k'm'}|^2=1$. For constant $|U^{(\ell)}_{jkm,j'k'm'}|^2=c$ the sum is $c g_\ell$, so $c=1/g_\ell$. 

Consider now the assumption {\bf (R)}.  For a random $g\times g$ unitary matrix $M(\omega)$ distributed according to the Haar measure, one has $\mathbb E[|M_{kl}(\omega)|^2]=1/g$. Therefore 
$\mathbb E[ |U^{(\ell)}_{jkm, j'k'm'}\big|^2] = 1/g_\ell$.

In either way | understanding that we take the average in the random case | we have from \eqref{12},
\begin{align}
\label{14}
{\rm Tr}\big(U(\rho\otimes\tau_{\beta}) U^* &H_\a\big)  = \sum_\ell \frac{1}{g_\ell} \sum_{jj'kk'} \chi\big(\ve_j+\eta_k=\ve_{j'}+\eta_{k'}=E_\ell\big)\ve_j p_{j'} \frac{m^2_\b e^{-\beta\eta_{k'}}}{Z_\b}= \vec\ve\cdot G\vec p,
\end{align}
where the factor $m^2_\b$ is due to the sums over $m,m'$ and where $G$ is the $d\times d$ matrix with elements
\begin{equation}
\label{matG}
{G}_{jj'} = \sum_\ell \frac{1}{g_\ell} \sum_{kk'} \chi\big(\ve_j+\eta_k=\ve_{j'}+\eta_{k'}=E_\ell\big)   \frac{m_\b^2e^{-\beta \eta_{k'}}}{Z_{\b}}.
\end{equation}
Here, $\chi$ is the characteristic function taking the value $1$ if the condition in the argument is satisfied and $0$ else. This shows \eqref{flow} for $G$ as given in \eqref{matG}. 

If in the above argument instead of taking the trace over $\h_\a$  (which is contained in the full trace $\rm Tr$) in \eqref{14} we consider $\langle \psi^\a_j| {\rm Tr}_\b\big(U(\rho\otimes\tau_{\beta}) U^* \big)|\psi^\a_j\rangle$, and we take $\rho=|\psi^\a_{j'}\rangle\langle\psi^\a_{j'}|$, then we obtain precisely that $\langle\psi^\a_j|\Phi_\beta(|\psi^\a_{j'}\rangle\langle \psi^\a_{j'})|\psi^\a_j\rangle
$ is given by \eqref{matG}, so \eqref{defGqCh} holds with $G_{jj'}$ given by \eqref{matG}.

Our next task is to see that the expressions \eqref{33} and \eqref{matG} are the same. We first observe that $G_{jj'}$ in \eqref{matG} is independent of the degeneracy $m_\b$: Both $g_\ell$ and $Z_\b$ are proportional to $m_\b$ and so they compensate the $m^2_\b$ in the numerator of \eqref{matG}. We conclude that we can set $m_\b=1$ in the following analysis of $G_{jj'}$. 
Introducing the integers $\ell_j$ defined by
\begin{equation}
\label{ellj}
\ve_s=\ve_1+\frac{\ell_s}{N},
\end{equation}
with $\ell_1=0$ and the convention that $\ell_{d+1}=\infty$, we write the sum over $\ell$ in \eqref{matG} as $\sum_\ell=\sum_{s=1}^d\sum_{\ell_s\leq \ell< \ell_{s+1}}$. Since in this sum (for $j,j'$ fixed) we have the constraint $E_\ell\ge \max\{\ve_j,\ve_{j'}\}$, we only need to sum over $s\ge  \max\{j,j'\}$. Also, for $\ell_s\leq \ell < \ell_{s+1}$, we have $g_\ell=s$, and $k, k'$ are uniquely determined by $E_\ell=\ve_j+\eta_k=\ve_{j'}+\eta_{k'}$, so we obtain from \eqref{matG},
\begin{align}
\label{21.1}
    G_{j j'}=\sum_{s=\max\{j,j'\}}^{d}\sum_{\ell_s\leq \ell < \ell_{s+1}}\frac{1}{s}\frac{e^{-\beta (E_\ell-\ve_{j'})}}{Z_{\b}}.
\end{align}
Now $E_\ell = \ve_1+\ell/N$ and (with the convention $\ve_{d+1}=\infty$)
\begin{equation}
\label{96}
\sum_{\ell_s\leq \ell < \ell_{s+1}} e^{-\beta E_\ell} = e^{-\beta \ve_1} \sum_{\ell_s\leq \ell < \ell_{s+1}} e^{-\beta\ell/N} =  e^{-\beta \ve_1} \frac{e^{-\beta\ell_s/N}-e^{-\beta \ell_{s+1}/N}}{1-e^{-\beta/N}} = \frac{e^{-\beta\ve_s} - e^{-\beta\ve_{s+1}}}{1-e^{-\beta/N}}.
\end{equation}
Using that $Z_B = (1-e^{-\beta/N})^{-1}$ and the expression \eqref{96}, we see that \eqref{21.1} equals \eqref{33}. This shows \eqref{33}. To complete the proof of  Proposition \ref{gjjp} we show that all eigenvalues of $G$ are in $(0,1]$.

The expression \eqref{33} shows that the stochastic matrix $G$ satisfies
$G_{jj'} >0$, $\forall j,j'$. The Perron-Frobenius theorem implies that $G$ admits a unique invariant vector and this vector has  strictly positive entries (up to normalization) and moreover, all other eigenvalues (besides the simple eigenvalue $1$) have modulus strictly smaller than $1$. The equation \eqref{invargibbs} shows the invariant vector is $v_\beta\propto (e^{-\beta\ve_1},\ldots,e^{-\beta\ve_d})$  (see also Lemma \ref{c-lem}). Next we show that all eigenvalues of $G$ are strictly positive. 
We check on \eqref{33} that $G^*D^{-1}= D^{-1}G$ holds, where $D=\text{diag}(e^{-\beta\ve_1},\ldots,e^{-\beta \ve_d})$. This is known as the detailed balance relation \cite{N} (see also Lemma \ref{lem3}(b) below). It follows that $\langle \varphi |D^{-1}G\psi\rangle=\langle G\varphi |D^{-1}\psi\rangle$ for all $\phi,\psi\in\h_A$, that is, $G$ is selfadjoint on $\h_\a$ with inner product $\langle \varphi|\psi\rangle_D:=\langle \varphi|D^{-1}\psi\rangle$. To show that $G$ has positive spectrum it thus suffices to show that $\langle\varphi, D^{-1}G\varphi\rangle>0$, $\forall\varphi\in\h_\a$. Now \eqref{33} gives 
\begin{align}\label{yyy}
G = \Gamma D^{-1},\qquad \text{where}\qquad \Gamma_{jj'} = \sum_{s=\max\{j,j'\}}^{d}\frac{1}{s}\big(e^{-\beta \ve_s}-e^{-\beta \ve_{s+1}}\big)\ \ \  \text{with \ }\ e^{-\beta \ve_{d+1}}=0.
\end{align}
Thus we have $D^{-1}G=D^{-1}\Gamma D^{-1}$ and hence the strict positive definiteness of $D^{-1}G$ and $\Gamma$ are equivalent. For $j=1,\ldots,d$, set (with $e^{-\beta \ve_{d+1}}=0$),
\begin{align}
\label{petitgj}
g_j=\sum_{s=j}^{d}\frac{1}{s}\big(e^{-\beta \ve_s}-e^{-\beta \ve_{s+1}}\big).
\end{align}
Then we have $g_1>g_2>\dots >g_d>0$ and
\begin{align}
\Gamma=\begin{pmatrix}
g_1 & g_2& g_3 & \dots & g_d \\
g_2 & g_2 & g_3 & \cdots & g_d\\
g_3 & g_3 & g_3 & \cdots & g_d \\
\vdots & \vdots & \vdots & \ddots & \vdots\\
g_d & g_d & g_d & \dots & g_d
\end{pmatrix} = TPT^*, \qquad T=\begin{pmatrix}
1 & 1& 1 & \dots & 1 \\
0 & 1 & 1 & \cdots & 1\\
0 & 0 & 1 & \cdots & 1 \\
\vdots & \vdots & \vdots & \ddots & \vdots\\
0 & 0 & 0 & \dots & 1
\end{pmatrix},
\end{align}
and where $P=\text{diag}(g_1-g_2, g_2-g_3, \dots, g_{d-1}-g_d, g_d)$ is a strictly positive definite matrix. $T$ is invertible so  $\Gamma$ is a strictly positive definite matrix. This shows that the eigenvalues of $G$ are strictly positive, thus the spectrum of $G$ is contained in $(0,1]$. The proof of Proposition \ref{gjjp} is complete.\hfill\qed

\bigskip

The stochastic matrix $G$ associated to the quantum channel $\Phi_\beta$ depends on the temperature $\beta$, so we write $G=G(\beta)$. We note the following properties: 
\medskip

(i) The matrix $G(\beta)$ is independent of the inverse gap size $N$ and of the multiplicity $m_\b\geq 1$ of the spectrum of $H_\b$ (so we may take $m_\b=1$), and is invariant under translations of $H_\a$ by $\alpha \bbbone$.

(ii) The matrix valued function $G(\beta)$ defined by \eqref{33}, admits an analytic extension to $\beta \in\mathbb C$, which we denote by the same symbol. 

(iii) In particular, $G(0)=\frac{1}{d}{\bf 1}{\bf 1}^T$, where $\bf 1$ is the vector with all components equal to one. The spectrum of $G(0)$ is thus $\{0,1\}$, where the multiplicity of $0$ equals $d-1$. Then we have $G\vec{p}=\frac{1}{d}{\bf 1}$ for any probability vector $\vec p$. Hence $G(0)\vec p$ corresponds to the infinite temperature Gibbs state $\bbbone/d$, for all $\vec p$. Moreover, for $\beta\rightarrow\infty$ we get
\begin{align}\label{ginfini}
G(\infty)=\begin{pmatrix}
1 & 1/2 & 1/3 & \dots & 1/d\\
        0 & 1/2 & 1/3 & \dots & 1/d\\
        0 & 0 & 1/3 & \dots & 1/d\\
        \vdots & \vdots & \ddots &\ddots & \vdots \\
        0 & 0 &\dots   & 0 & 1/d
\end{pmatrix},
\end{align}
with spectrum $\{1, 1/2, \dots, 1/d\}$, and invariant probability vector $\vec p =(1, 0, \dots, 0)^T$ corresponding to the ground state of $H_\a$.

(iv) When $\beta<0$ with $|\beta|$ large enough, we get that $G_{j j'}(-|\beta|)<0$, for $\max\{j, j'\}<d$. This shows that if $\Phi_{-|\beta|}$ has a well defined extension to $\beta< 0$, this extension will not be a positive map (recall \eqref{defGqCh}). In the same vein, for any $\beta<0$ we have $G_{j\, d-1}(-|\beta|)<0$ for $\ve_d-\ve_{d-1}$ large enough, which implies that $\Phi_{-|\beta|}$, if well defined, is not positive.

\section{Reduced thermal quantum channels}
\label{secRedQChannels}

Consider a general bipartite system with Hilbert space $\h_{\rm tot}=\h\otimes\h_\rr$ where  $\dim\h=d<\infty$ and $\dim\h_\rr\le \infty$. Let $H$ and $H_\rr$ be self-adjoint, bounded below Hamiltonians on $\h$ and $\h_\rr$, given by
\begin{equation}
\label{hamilts}
H=\sum_j \ve_j|\psi_j\rangle\langle\psi_j|,\qquad H_\rr=\sum_k \eta_k |\chi_k\rangle \langle \chi_k|,
\end{equation}
with finitely degenerate eigenvalues, and where the eigenvectors are orthonormal bases of the respective spaces.  The equilibrium states at inverse temperature $\beta\in \mathbb R$ are 
$$
\gamma_\beta = \frac{e^{-\beta H}}{Z_\beta},\qquad \tau_\beta = \frac{e^{-\beta H_\rr}}{Z_{\beta, \rr}}.
$$
In the infinite dimensional case we assume that $e^{-\beta H_\rr}$ has a finite trace, which necessitates $\beta>0$ since $H_\rr$ is bounded below. Set 
$$
H_{\rm tot}=H+H_\rr
$$ 
where we omit trivial tensor factors in the notation (and don't fret about the simple domain considerations). Let $V$ be a unitary on $\h_{\rm tot}$ satisfying
\begin{align}\label{comV}
    [H_{\rm tot},V]=0.
\end{align}
We define the reduced thermal quantum channel $\Phi: \mathcal B(\h)\rightarrow\mathcal B(\h)$ by
\begin{equation}
\label{defphibeta'}
    \Phi (X)=\tr_\rr\big(V(X\otimes \tau_\beta)V^*\big),
\end{equation}
where $\tr_\rr$ is the partial trace over $\h_\rr$. The map $\Phi$ is CPTP (completely positive, trace preserving)  \cite{H, EH-K, Kr, S, NMlectnotes} and has $\gamma_\beta$ as an invariant vector. Complete positivity implies directly that  for all $X\in\mathcal B(\h)$,
\begin{equation}
\label{14.1}
\Phi(X^*) = \big[\Phi(X)\big]^*.
\end{equation}
The spectrum of the map $\Phi$ is thus invariant under complex conjugation: the eigenvalues of $\Phi$ are either real or they appear in complex conjugate pairs. 
Moreover, again by virtue of $\Phi$ being CPTP, its spectrum $\sigma(\Phi)$ satisfies
\begin{align}
    \sigma(\Phi)\subset \{z\in \mathbb C \ | \ |z|\leq 1\} \quad \text{and} \quad 1\in \sigma(\Phi).
\end{align}

\medskip

Denote the sets of operators on $\h$ which are diagonal and off diagonal relative to the eigenbasis $\{\psi_j\}_{j=1}^d$ of $H$, by
\begin{align}
\mathcal D &=\big\{X\in \mathcal B(\h) \ \text{s.t. $\langle \psi_j|X\psi_k \rangle =0$ for $j\neq k$}\big\},\nonumber\\
\mathcal O &=\big\{X\in \mathcal B({\cal H}) \ \text{s.t. $\langle\psi_j|X\psi_j\rangle =0$ for all $j$}\big\}.
\end{align}
These two sets are orthogonal complements of each other when considering $\mathcal B(\h)$ as a Hilbert space with inner product $\langle X|Y\rangle_{\mathcal B}=\tr(X^*Y)$, and we have the orthogonal decomposition
\begin{equation}
\label{bhdec}
\mathcal B(\h) = \mathcal D\oplus \mathcal O.
\end{equation}

\begin{lem}
\label{phiprop}
Suppose that all eigenvalues $\ve_j$ of $H$ are simple. Then
\begin{itemize}
\item[\bf (a)] $\Phi$ leaves both  $\mathcal D$ and $\mathcal O$ invariant, that is, $\Phi$ is reduced by (block-diagonal in) the decomposition \eqref{bhdec}.

\item[\bf (b)] If all nonzero Bohr energies $\ve_j-\ve_k$ of $H$ are distinct, then $|\psi_j\rangle\langle\psi_k|$ is an eigenvector of $\Phi$ whenever $j\neq k$, with eigenvalue $z_{j,k}$ such that $|z_{j,k}|\le 1$. Moreover, $z_{j,k}=\overline{z_{k,j}}$.
\end{itemize}
\end{lem}

\noindent
{\bf Proof of Lemma \ref{phiprop}.} 
We start the proof by noticing that the map $\,{\cal U}_t:\rm  {\cal B}({\cal H})\rightarrow {\cal B}({\cal H})$ given by $\,{\cal U}_t(X)=e^{-itH}Xe^{itH}$ satisfies $\Phi \circ {\cal U}_t={\cal U}_t\circ \Phi$ for all $t\in \mathbb R$. We provide the argument for the convenience of the reader (see also the review article \cite{Lo}). For any $X\in {\cal B}({\cal H})$ we have
\begin{align}
\Phi ({\cal U}_t(X))&=\tr_\rr(V(e^{-itH}Xe^{itH}\otimes \tau_\beta)V^*)\nonumber\\
&=\tr_\rr(V (e^{-itH}Xe^{itH}\otimes e^{-itH_\rr}\tau_\beta e^{itH_\rr} )V^*)\nonumber\\
&=\tr_\rr(V e^{-itH_{\rm tot}}(X\otimes \tau_\beta) e^{itH_{\rm tot}} V^*)\nonumber\\
&=\tr_\rr(e^{-itH_{\rm tot}}V (X\otimes \tau_\beta)V^* e^{itH_{\rm tot}})\nonumber\\
&=e^{-itH} \tr_\rr(e^{-itH_\rr} V (X\otimes \tau_\beta)V^*e^{itH_\rr}) e^{itH}\nonumber\\
&=e^{-itH} \tr_\rr( V (X\otimes \tau_\beta)V^*) e^{itH}
\nonumber\\
&={\cal U}_t(\Phi (X)).
\label{16}
\end{align}
$\mathcal U_t$ is a unitary group on $\mathcal B$ and its selfadjoint generator $\mathcal L$ is the commutator $\mathcal L = [H,\cdot\,]$, that is, $\mathcal U_t=e^{-it\mathcal L}$. The eigenvalues of $\mathcal L$ are the Bohr energies $\Delta=\ve_i-\ve_j$. They have multiplicity $m(\Delta)=|\{(i,j)\ :\ \ve_i-\ve_j=\Delta\}|$. For every $(i,j)$ the element  $X_{ij}:=|\psi_i\rangle\langle\psi_j|\in\mathcal B(\h_\a)$ is a normalized eigenvector of $\mathcal L$ with eigenvalue $\ve_i-\ve_j$. The set $\{X_{ij}\ : \ \ve_i-\ve_j=\Delta\}$ is an orthonormal basis for the eigenspace of $\mathcal L$ associated to the eigenvalue $\Delta$. Therefore the eigenprojection of $\mathcal L$ associated to $\Delta$ is given by
$$
\Pi_\Delta = \sum_{i,j\, :\, \ve_i-\ve_j=\Delta} |X_{ij}\rangle_{\mathcal B}\langle X_{ij}|,
$$
where for $X,Y\in\mathcal B(\h)$ we define $|X\rangle_{\mathcal B}\langle Y| : \mathcal B(\h)\rightarrow\mathcal B(\h)$ by the action $|X\rangle_{\mathcal B}\langle Y| (Z)=X\,\tr(Y^*Z)$.
The spectral decomposition of $\mathcal L$ is then,
$$
\mathcal L = \sum_\Delta \Delta \, \Pi_\Delta, 
$$
where the sum runs over all (distinct) Bohr energies. By \eqref{16}, $\Phi$ and $\mathcal L$ commute, which is equivalent to saying that $\Phi$ and $\Pi_\Delta$ commute for all $\Delta$. In other words, $\Phi$ leaves each eigenspace of $\mathcal L$ invariant.

So far we did not use that the $\ve_j$ are all distinct. Now if they are,  then the eigenspace associated to $\Delta=0$ is spanned by $\{X_{ii}\}_{i=1}^d$, which is precisely $\mathcal D$. This shows that $\Phi$ leaves $\mathcal D$ invariant. Next set $Q_0=\sum_{\Delta \neq 0}\Pi_\Delta$, so that $Q_0=\bbbone-\Pi_0=(\Pi_0)^\perp$ is the projection onto the orthogonal complement of $\mathcal D$, which is exactly $\mathcal O$. As $\Phi$ commutes with each $\Pi_\Delta$ it commutes with $Q_0$. Thus $\Phi$ leaves $\mathcal O$ invariant. This shows (a).

Now suppose that the nonzero eigenvalues $\Delta$ of $\mathcal L$ are all nondegenerate, then the associated spectral subspaces have dimension one and are spanned by a single $X_{ij}=|\psi_i\rangle\langle\psi_j|$. Hence $\Phi$ maps each $|\psi_i\rangle\langle\psi_j|$ into a multiple of itself ($i\neq j$).

Let $z_{i,j}$ be the eigenvalue associated to $|\psi_i\rangle\langle\psi_j|$, that is, $\Phi(|\psi_i\rangle\langle\psi_j|) = z_{i,j}|\psi_i\rangle\langle\psi_j|$. Then $|z_{i,j}|\leq 1$ and $z_{i,j}=\overline{z_{j,i}}$ follows from \eqref{14.1}.
This completes the proof of the Lemma \ref{phiprop}.   \hfill\qed

\medskip

The restriction of the channel $\Phi$ to diagonal operators has the following physical meaning. If the system starts out in the eigenstate $\psi_{j'}$ before interacting with $\rr$, then the probability of finding the system (upon a quantum measurement of the observable $H$) in the eigenstate $\psi_j$ after the interaction with $\rr$ is 
\begin{equation}
\label{Gabs}
G_{jj'}=\big[\tr_\rr(V(\ket{\psi_{j'}}\bra{\psi_{j'}}\otimes \tau_{\beta})V^*)\big]_{jj} = \big[\Phi(\ket{\psi_{j'}}\bra{\psi_{j'}})\big]_{jj},
\end{equation}
where we use the notation  $[\cdots]_{kl}=\langle\psi_k|\cdots|\psi_l\rangle$. The relation \eqref{Gabs} defines a $d\times d$ matrix $G$ whose $(j,j')$ entries are probabilities for the transition $\psi_{j'}\rightarrow \psi_j$ induced by the contact with $\rr$. The matrix represents the action of $\Phi$ on $\mathcal D$ which is identified as $\mathcal D \simeq \{ v\in \mathbb C^d\}$. For any given $j'$ the probabilities for the transitions $j'\rightarrow j$ add up to one when we sum over $j$, so
\begin{equation}
\label{stoch1}
\sum_{j=1}^{d}G_{jj'}=1,\quad \forall j'.
\end{equation}
$G$ is the transition matrix of a Markov chain, or a stochastic matrix | characterized by having nonnegative entries such that the sum of each column adds up to $1$, see also \cite{Lo}.  The constant vector $v_0=(1,1,\ldots,1)^T\in{\mathbb C}^d$ represents the equilibrium state $\gamma_\beta$ at $\beta=0$. As $(G^*)_{j'j}=G_{jj'}$, \eqref{stoch1} means that 
\begin{equation}
\label{ev0}
G^*v_0=v_0.
\end{equation}
By linearity of \eqref{Gabs} in the system density matrix and the fact that $V$ commutes with $H+H_\rr$ we obtain
\begin{align}
\sum_{j'=1}^{d}G_{jj'}\frac{e^{-\beta \ve_{j'}}}{Z_\beta} & =\big[ \tr_\rr(V(\gamma_{\beta}\otimes \tau_{\beta})V^*)\big]_{jj}=\frac{1}{Z_\beta Z_{\rr,\beta}}\big[\tr_\rr(V e^{-\beta (H+H_\rr)} V^*)\big]_{jj}\nonumber\\
&=  \frac{1}{Z_\beta Z_{\rr,\beta}}\big[\tr_\rr(e^{-\beta (H+H_\rr)})\big]_{jj}=\frac{e^{-\beta \ve_{j}}}{Z_\beta},
\end{align}
which means that $G$ leaves the system Gibbs state $\rho_\beta$ invariant, in the sense that 
\begin{equation}
\label{ev1}
Gv_\beta=v_\beta
\end{equation}
where $v_\beta$ is the vector with $j$th entry equal to $Z_\beta^{-1}e^{-\beta \ve_j}$. (This property is well known \cite{Lo}.) 

Let us consider the adjoint $\Phi^*$ of $\Phi$, where $\Phi$ is considered as a linear operator on $\mathcal B(\h)$ with inner product $\langle X,Y\rangle_{\mathcal B} = {\rm Tr}(X^*Y)$. From definition \eqref{Gabs}, 
\begin{align}\label{Gadj}
   (G^*)_{kl} &= \overline{G_{lk}}= \overline{\big\langle \ket{\psi_{l}}\bra{\psi_{l}} , \Phi(\ket{\psi_{k}}\bra{\psi_{k}})\big \rangle_{\mathcal B}}= \big\langle \ket{\psi_{k}}\bra{\psi_{k}} , \Phi^*(\ket{\psi_{l}}\bra{\psi_{l}})\big \rangle_{\mathcal B}\nonumber\\
   &=\big[\Phi^*(\ket{\psi_{l}}\bra{\psi_{l}})\big]_{kk}.
\end{align}
Replacing $V$ by $V^*$ in the definition of $\Phi$, \eqref{defphibeta'}, gives a second CPTP map on $\Psi$ on $\mathcal B(\h)$, defined by 
\begin{equation}
\label{defphibeta''}
    \Psi (X)=\tr_\rr\big(V^*(X\otimes \tau_\beta)V\big),
\end{equation}
for which all the above considerations hold ($\Psi$ leaves $\gamma_\beta$ invariant, Lemma \ref{phiprop} holds for $\Psi$, etc.). 
We have for any $X,Y\in\mathcal B(\h)$,
\begin{align*}
\langle X,\Phi(Y)\rangle_{\mathcal B} &= \tr_{\h_{\rm tot}} \big[ (X^*\otimes \bbbone_\rr) V(Y\otimes\tau_\beta)V^*\big]\nonumber\\
&=\tr_{\h_{\rm tot}} \big[V^*(X^*\otimes \bbbone_\rr) V (e^{-\beta H}\otimes\tau_\beta) (e^{\beta H}Y\otimes\bbbone_\rr)\big]\nonumber\\
&=\tr_{\h_{\rm tot}} \big[V^*((e^{-\beta H}X)^*\otimes \tau_\beta) V (e^{\beta H}Y\otimes\bbbone_\rr)\big]\nonumber\\
&=\tr_{\h} \big( [\Psi(e^{-\beta H}X)]^* e^{\beta H}Y\big)\nonumber\\
&=\langle e^{\beta H}\Psi(e^{-\beta H}X), Y\rangle_{\mathcal B}.
\end{align*}
We conclude that 
\begin{equation}
\label{phistarpsi}
\Phi^*(X) = e^{\beta H}\Psi\big(e^{-\beta H}X\big),\qquad \forall X\in\mathcal B(\h).
\end{equation}
It is then clear that $\Phi^*$ and $\Psi$ have the same spectrum. 
Just as with $\Phi$ and $G$, associated to the channel $\Psi$ is the stochastic matrix F with matrix elements ({\it c.f.}~\eqref{Gabs}) 
\begin{equation}
\label{Gabs''}
F_{jj'}=\big[\tr_\rr(V^*(\ket{\psi_{j'}}\bra{\psi_{j'}}\otimes \tau_{\beta})V)\big]_{jj} = \big[\Psi(\ket{\psi_{j'}}\bra{\psi_{j'}})\big]_{jj}.
\end{equation}
Then, relations \eqref{phistarpsi} with $X=\ket{\psi_{l}}\bra{\psi_{l}}$ and \eqref{Gadj} yield $e^{-\beta \ve_k} (G^*)_{kl} = F_{kl}e^{-\beta \ve_l}$, or in matrix form,
\begin{equation}
\label{23}
G^*= D^{-1}F D, \qquad \text{where} \qquad D=\text{diag}(e^{-\beta\ve_1},\ldots,e^{-\beta \ve_d}).
\end{equation}

By Lemma \ref{phiprop}, $\Phi$ leaves both $\mathcal D$ and $\mathcal O$ invariant. We denote the restrictions by $\Phi|_{\mathcal D}$ and $\Phi|_{\mathcal O}$. The eigenvector $v_\beta\propto (e^{-\beta\ve_1},\ldots,e^{-\beta\ve_d})$ of $G$ with eigenvalue $1$ (c.f. \eqref{ev1}) has strictly positive entries. Thus, if $G_{jj'}>0$ for all $j,j'$, then it follows from the Perron-Frobenius theorem that $v_\beta$ is the Perron-Frobenius eigenvector (the unique eigenvector with  all strictly positive entries) and $1$ is the Perron-Frobenius eigenvalue, having the property that it is a simple eigenvalue and all other eigenvalues have modulus smaller than $1$. The same applies to $\Psi$ and $F$. This shows the assertion (a) of the following result.

\begin{lem}
\label{c-lem}
Suppose that all eigenvalues $\ve_j$ of $H$ are simple and that $G_{jj'}>0$ for all $j,j'$ (respectively $F_{jj'}>0$ for all $j,j'$). Then
\begin{itemize}
\item[{\rm \bf (a)}]
$1$ is a simple eigenvalue of $\Phi|_{\mathcal D}$ (respectively $\Psi|_{\mathcal D}$) with associated eigenvector $v_\beta$ (representing $\rho_\beta$ as a vector in $\mathbb C^d$). All other eigenvalues have modulus less than $1$. 

\item[{\rm \bf (b)}] If in addition all nonzero Bohr energies $\ve_j-\ve_k$ of $H$ are distinct,  then the spectrum of $\Phi|_{\mathcal O}$ (respectively that of $\Psi|_{\mathcal O}$) lies in the open complex unit disk.
\end{itemize}
\end{lem}

The property (b) of  Lemma \ref{c-lem} means that the energy coherences (off-diagonal density matrix elements) of a density matrix $\rho$ cannot increase when $\Phi$ is applied. Since $\Phi=\Phi|_{\mathcal D}\oplus\Phi|_{\mathcal O}$, Lemma \ref{c-lem} immediately implies the following result.

\begin{cor}
\label{cor1}
Suppose that all eigenvalues $\ve_j$ of $H$ are simple. If $G_{jj'}>0$ for all $j,j'$ (respectively $F_{jj'}>0$ for  all $j,j'$) and the Bohr energies $\ve_j-\ve_k$ of $H$ are all distinct, then $\Phi$ (respectively $\Psi$) has a simple eigenvalue $1$ with eigenvector $\rho_\beta$ and all other eigenvalues have modulus less than $1$. 
\end{cor}

{\bf Proof of Lemma \ref{c-lem}.} (a) was proven before the lemma. To prove (b) we note that by Lemma \ref{phiprop}(b) the eigenvalues of $\Phi|_{\mathcal O}$ are $z_{k,l}=\big\langle \psi_k, \Phi(|\psi_k\rangle\langle\psi_l|)\psi_l\big\rangle
$, for $k\neq l$, and they are simple. By the definition \eqref{defphibeta'} of $\Phi$ and using the basis $|\chi_k\rangle$ of $\h_\rr$ (see \eqref{hamilts}), we obtain
\begin{equation*}
z_{k,l} = \sum_{r,s} \frac{e^{-\beta\eta_r}}{Z_{\beta,\rr}} V_{ks,kr} \overline{V_{ls,lr}},
\end{equation*}
where $V_{ks,k's'} = \langle \psi_k\otimes\chi_s|V\,\psi_{k'}\otimes\chi_{s'}\rangle$. It follows from the Cauchy-Schwarz inequality that 
\begin{equation}
|z_{k,l}|^2 \le \Big( \sum_{r,s} \frac{e^{-\beta\eta_r}}{Z_{\beta,\rr}} |V_{ks,kr}|^2 \Big) \Big( \sum_{r,s} \frac{e^{-\beta\eta_r}}{Z_{\beta,\rr}} |V_{ls,lr}|^2\Big). 
\label{30}
\end{equation}
Again by \eqref{defphibeta'} and using the $|\chi_k\rangle$ to express the trace over $\h_\rr$ we obtain for the  diagonal matrix elements of $G$, \eqref{Gabs},
\begin{equation}
G_{jj} = \sum_{r,s}\frac{e^{-\beta \eta_r}}{Z_{\beta,\rr}} |V_{js,jr}|^2,
\end{equation}
so that \eqref{30} becomes $|z_{k,l}|^2\le G_{kk}G_{ll}$. Now if $|z_{k,l}|=1$ then $G_{kk}=G_{ll}=1$, which means that the $k$th and the $l$th column of the stochastic matrix $G$ are just the $k$th and $l$th canonical basis vector. Hence those vectors are both eigenvectors of $G$ with eigenvalue 1. This contradicts the simplicity of the eigenvalue $1$ of $G$. It follows that $|z_{k,l}|<1$ and the proof of Lemma \ref{c-lem} is complete.\hfill\qed
\smallskip

{\bf Remark.} Lemma \ref{c-lem}  holds under the weaker assumption that the stochastic matrix $G$ is irreducible and aperiodic.  {For the applications that follow, the property $G_{jj'}>0$ for all $j,j'$ holds naturally, hence our simpler formulation. }

\medskip

Recall the definitions of the channels $\Phi$ and $\Psi$ in \eqref{14.1} and \eqref{defphibeta''}. We say that the channel $\Phi$, \eqref{14.1} is {\it homogeneous} if $\Phi|_{\cal D}=\Psi|_{\cal D}$, that is, 
\begin{equation}
\label{homog}
\tr_\rr\big(V(X\otimes \tau_\beta)V^*\big) = \tr_\rr\big(V^*(X\otimes \tau_\beta)V\big),\qquad \forall X\in\cal D\subset \mathcal B(\h). 
\end{equation}
This relation is equivalent to $G=F$, the equality of the stochastic matrices \eqref{Gabs}, \eqref{Gabs''} representing the channels reduced to $\mathcal D$. 

\begin{lem}
\label{lem3}
\ \ 
\begin{itemize}
\item[\rm (a)] Suppose that $\Phi$ is a homogeneous channel with associated stochastic matrix $G$ satisfying $G_{jj'}>0$ for all $j,j'$. Then the spectrum of $\Phi|_{\cal D}$ is a subset of the interval $(-1,1]$.

\item[\rm (b)] The channel $\Phi$ is homogeneous if $V$ only allows energy conserving transitions, and with equal probability at any fixed energy, that is if
$$
|V_{jr,j's}|^2 = \left\{
\begin{array}{cl}
p_E & \text{if $\ve_j+\eta_r = \ve_{j'}+\eta_s\equiv E$}\\
0 & \text{if $\ve_j+\eta_r \neq  \ve_{j'}+\eta_s$}
\end{array}
\right.
$$
for all total energies $E$ and some $p_E$.
\end{itemize}
\end{lem}

{\bf Proof of Lemma \ref{lem3}.} 
(a)  For a homogeneous channel we have $G=F$, so \eqref{23} implies the intertwining relation $D^{-1}G = G^*D^{-1}$, or detailed balance condition, see \cite{N}. This implies that $G$ is selfadjoint as an operator on the Hilbert space $\h_D:=(\mathbb C^d,\langle \cdot | \cdot \rangle_D)$, were the new inner product is 
\begin{equation}
\label{newip}
\langle \varphi|\psi\rangle_D:=\langle \varphi|D^{-1}\psi\rangle.
\end{equation}
Indeed, we have
\begin{align*}
    \langle \varphi | G \psi\rangle_D=\langle \varphi |D^{-1}G\psi\rangle=\langle \varphi |G^*D^{-1}\psi\rangle=\langle G\varphi |D^{-1}\psi\rangle=\langle G\varphi |  \psi\rangle_D.
\end{align*}
Hence the spectrum of $G$ is real. 
We already know from Lemma \ref{c-lem} that $1$ is a simple eigenvalue of $G$ and all other eigenvalues have modulus strictly less than $1$. 

(b) Using the definition of $\gamma_\beta$ and expanding the trace over $\h_\rr$ in the basis $\chi_k$ of $H_\rr$, \eqref{hamilts}, we obtain for the matrix elements,
\begin{equation}\label{matelGF}
G_{jj'} = \sum_{r,s}\frac{e^{-\beta \eta_s}}{Z_{\beta,\rr}} |V_{jr,j's}|^2,\qquad F_{jj'} = \sum_{r,s}\frac{e^{-\beta \eta_s}}{Z_{\beta,\rr}} |V_{j's,jr}|^2,
\end{equation}
where $V_{ab,cd}=\langle\psi_a\otimes\chi_b, V\psi_c\otimes\chi_d\rangle$. Under the condition on $|V_{jr,j's}|^2$ in the Lemma we have $|V_{jr,j's}|^2=|V_{j's,jr}|^2$, so $G=F$. This concludes the proof of Lemma \ref{lem3}.\hfill\qed

\bigskip

We close this Section with another general structural property of $\Phi$, the monotony with respect to the free energy functional.  For $\beta\in \mathbb R$, let ${\cal F}_\beta$ be the real valued continuous map defined on the compact set of density matrices on $\cal H$ by
\begin{align}
{\cal F}_\beta(\rho)=\beta \tr(\rho H)-S(\rho),
\end{align}
where $S(\rho)=-\tr (\rho\ln \rho)$.
The functional ${\cal F}_\beta$ is the free energy of $\rho$ multiplied by $\beta$. The relative entropy of $\rho$ with respect to $\sigma$, two density matrices s.t. $\sigma>0$, is given by $S(\rho|\sigma)=\tr(\rho \ln \rho-\rho \ln \sigma)$. By Klein's inequality, $S(\rho|\sigma)\geq 0$ with equality if and only if $\rho=\sigma$ \cite{C}.

\begin{lem}\label{lem4}
Under the assumption \eqref{comV} and for $\beta\in\mathbb R$ such that $\Phi$ given in \eqref{defphibeta'} is positive, we have for any density matrix $\rho$, 
\begin{align}
\label{73}
{\cal F}_\beta(\Phi(\rho))\leq {\cal F}_\beta(\rho).
\end{align}
\end{lem}
\begin{proof}
By the data processing inequality, valid for positive trace-preserving $\Phi$, \cite{Li, U, M-HR}, and using that $\gamma_\beta$ is invariant under $\Phi$ we have,
\begin{align}
\label{74}
S(\Phi(\rho)|\Phi(\gamma_\beta))=S(\Phi(\rho)|\gamma_\beta)\leq S(\rho | \gamma_\beta).
\end{align}
Next, $S(\rho | \gamma_\beta)=-S(\rho)+\beta \tr(\rho H)+\ln(Z_\beta)={\cal F}_\beta(\rho)+\ln(Z_\beta)$ and the same holds for $\rho$ replaced by $\Phi(\rho)$, namely $S(\Phi(\rho) | \gamma_\beta)={\cal F}_\beta(\Phi(\rho))+\ln(Z_\beta)$. Then \eqref{74} implies \eqref{73}.
\end{proof}

Lemma \ref{lem4} provides a link between the energy the entropy variations under the action $\Phi$, known as Landauer's bound: 
\begin{align}\label{lb}
    \beta\big[\tr (\Phi(\rho)H)-\tr(\rho H)\big]\leq S(\Phi(\rho))-S(\rho).
\end{align}

Coming back to our thermometer model of Section \ref{sec:model} we recognize the left side of the inequality \eqref{lb} as the flow $Q_\a(\beta)$ (see \eqref{EvarA}) times $\beta$, that is we have,
\begin{equation}
\label{76}
\beta Q_\a(\beta) \le S(\Phi_\beta(\rho))-S(\rho).
\end{equation}
Suppose $\rho$ is such that $Q_\a(0)>0$. Then $\bop>0$. We know that $\Phi_\beta$ is CPTP for $\beta>0$, so \eqref{76} holds by Lemma \ref{lem4}, and it shows that (recall that $Q_\a(\beta)$ is monotone decreasing)
\begin{align}
\label{entrobop}
S(\Phi_\beta(\rho))\geq S(\rho) \ \ \text{for all}  \ \ 0<\beta\leq \bop.
\end{align}
We conclude that the von Neumann entropy can only increase under the application of the thermal channel $\Phi_\beta$, in the given range of parameters.
\medskip

{\it Remark.} For finite-dimensional $\h_\a,\h_\b$ a more refined argument shows the following result. For $\rho$ such that $Q_\a(0)>0$, we have $S(\Phi(\rho))-S(\rho)>0$ for all $0<\beta<\bop$ and $S(\Phi(\rho))-S(\rho)\geq 0$ if $\beta=\bop$. (We do not present the proof here.)

\subsection{Application: Approach of equilibrium}
\label{sec:appreq}

The energy flow of the thermometer model of Section \ref{sec:model} is given by \eqref{EvarA}, 
\begin{equation}
Q_\a(\beta) = \tr\big(\Phi_\beta(\rho) H_\a \big)-\tr(\rho H_\a),
\end{equation}
where $\Phi_\beta$, defined in \eqref{defPhibeta}, is a reduced thermal quantum channel. The associated stochastic matrix has the expression \eqref{33} of Proposition \ref{gjjp}, from which we see explicitly that $G_{jj'}>0$ for all $j,j'$. Lemma \ref{lem3}(b) shows that $\Phi_\beta$ is a homogeneous channel, because $U$ only allows for energy conserving transitions with equal probability at fixed energy | this is the assumption (G). As a consequence of Lemmas \ref{phiprop}, \ref{c-lem} and Corollary \ref{cor1} we obtain the following result.

\begin{thm}
\label{applithm}
Let $\Phi_\beta$,  \eqref{defPhibeta}, be the reduced thermal quantum channel of the thermometer model of Section \ref{sec:model}.  Then 
\begin{itemize}
\item[{\rm (a)}] $\Phi_\beta$ leaves the space of diagonal matrices (w.r.t. $H_\a$) invariant, and it leaves the space of off-diagonal matrices invariant: $\Phi_\beta=\Phi_\beta|_{\mathcal D}\oplus\Phi_\beta|_{\mathcal O}$. 

\item[{\rm (b)}] If all nonzero Bohr energies $\ve_j-\ve_k$ of $H_\a$ are distinct, then $\Phi_\beta|_{\mathcal D}$ has a simple eigenvalue $1$ with associated spectral projection $|v_\beta\rangle\langle v_0|$ (with $\mathcal D$ identified with $\mathbb C^d$ | see also \eqref{ev0}, \eqref{ev1}) and all other eigenvalues of $\Phi_\beta|_{\mathcal D}$, as well as all eigenvalues of $\Phi_\beta|_{\mathcal O}$, have modulus strictly less than $1$. 
\end{itemize}
\end{thm}

We may let the system in an initial state $\rho$ interact with the thermometer, then decouple it, resulting in the state $\Phi_\beta(\rho)$, and letting it subsequently interact with a fresh thermometer (at the same temperature as the previous one). Repeating this process $k$ times gives a discrete dynamics on the system: 
\begin{equation}
\rho\mapsto \rho^{(k)}\equiv \Phi_\beta^k(\rho),\qquad k=1,2,\ldots
\end{equation}
(apply $k$ times the map $\Phi_\beta$ to $\rho$). If the nonzero Bohr energies are distinct, then according to Theorem \ref{applithm}(b), we have for any density matrix $\rho$,
\begin{equation}
\label{rte}
\lim_{k\rightarrow\infty}\rho^{(k)}  =\gamma_\beta,
\end{equation}
where the speed of convergence is exponentially fast in $k$. The relation \eqref{rte} shows that a repeated interaction with fresh thermometers, all at a given temperature $\beta$, drives the system to the equilibrium state $\gamma_\beta$. For $k$ sufficiently large, \eqref{rte} ensures that $\rho^{(k)}$ is not the ground nor the most excited state of $H_\a$. Thus by Theorem \ref{thm2}, $\rho^{(k)}$ has a well defined contact temperature $\bop^{(k)}\in\mathbb R$. It satisfies (see \eqref{flow}),
\begin{equation}
\label{51.1}
Q_\a (\bop^{(k)}) = \vve \cdot G \vec{p}^{\, (k)} -\vve\cdot\vec{p}^{\, (k)} =0,
\end{equation}
where $\vec{p}^{\, (k)} = (p^{(k)}_1,\ldots,p^{(k)}_d)$ is the vector of the populations of $\rho^{(k)}$, with $p^{(k)}_j = \langle\psi^\a_j|\rho^{(k)}\psi^\a_j\rangle$. Of course, $G$ and the $\vec p^{\, (k)}$ on the right side of \eqref{51.1} depend on $\beta$ and the solution of the equation gives $\bop^{(k)}$. By \eqref{rte}, we have $p_j^{(k)}\rightarrow e^{-\beta \ve_j}/Z_\a$ as $k\rightarrow\infty$. 

Consider now the flow $Q_\a$ as a function of two variables, $\beta>0$ and $R\in\mathbb C^{d^2}$ (the matrix elements of an operator acting on $\mathbb C^d$). An explicit formula is given in \eqref{qabeta'}, for example. Suppose we have $Q_\a(\beta_0,\rho_0)=0$ for some $\beta_0\in\mathbb R$ and some density matrix $\rho_0$ (viewed as an element in $\mathbb C^{d^2}$). Since $\partial_\beta Q_
\a(\beta_0,\rho_0)\neq 0$ (this function is strictly monotonic, Theorem \ref{thm2}), the implicit function theorem tells us that there is a unique map $R\mapsto\beta(R)$, defined and differentiable in a neighbourhood of $\rho_0$, with $\beta(\rho_0)=\beta_0$, such that $Q_\a(\beta(R),R)=0$. Now let $\rho_0=\gamma_\beta$. We have $Q_
\a(\beta,\gamma_\beta)=0$ (see the remarks after Theorem \ref{thm2}), and for $k$ large, $\rho^{(k)}$ is close to $\gamma_\beta$ due to \eqref{rte}. By definition, $Q_\a(\bop^{(k)},\rho^{(k)})=0$. Therefore by uniqueness, $\beta(\rho^{(k)})=\bop^{(k)}$ for large $k$ and in particular, 
$$
\lim_{k\rightarrow\infty} \bop^{(k)} = \beta.
$$

\noindent {\bf Acknowledegments.} This work (A.J.) is partially supported by the French National Research Agency, grant Dynacqus ANR-24-CE40-5714-02. M.M.~was supported by a Discovery Grant from NSERC, the Natural Sciences and Engineering Reserach Council of Canada. Our thanks go to both our institutions for hosting mutual research visits. We thank D.~Spehner for useful discussions.

\appendix

\section{Proofs \ref{Qunbounded}}

\subsection{Proof of Theorem \ref{Qunbounded}}
\label{sec:proothm1}
We combine \eqref{flow} with \eqref{33},
\begin{align}
\label{38}
    Q_A=\sum_{j j'=1}^{d}\ve_j\sum_{s=\max\{j,j'\}}^{d}\frac{1}{s}\big(e^{-\beta (\ve_s-\ve_{j'})}-e^{-\beta (\ve_{s+1}-\ve_{j'})}\big)p_{j'}-\sum_{j=1}^{d} \ve_j p_j.
\end{align}
The right sides of \eqref{38} and of \eqref{qabeta} are entire analytic functions in $\beta\in\mathbb C$, which we call $F_1(\beta)$ and $F_2(\beta)$, respectively. To prove Theorem \ref{Qunbounded} it suffices to show that $\partial_\beta^\nu|_{\beta=0}F_1(\beta) =\partial_\beta^\nu|_{\beta=0}F_2(\beta)$ for all $\nu=0,1,2,\ldots$ We have for $\nu\ge 1$
\begin{equation}
\label{013}
\partial_\beta^\nu|_{\beta=0}F_1(\beta) = (-1)^\nu \sum_{jj'=1}^d \ve_j p_{j'} \sum_{k=M_{jj'}}^d\frac{(\ve_k-\ve_{j'})^\nu - (\ve_{k+1}-\ve_{j'})^\nu}{k}
\end{equation}
where $M_{jj'}=\max\{j,j'\}$. Our convention here is that if $k=d$ then $(\ve_{k+1}-\ve_{j'})^\nu=0$. The sum over $k$ is telescopic,
\begin{align}
\label{014}
\sum_{k=M_{jj'}}^d&\frac{(\ve_k-\ve_{j'})^\nu - (\ve_{k+1}-\ve_{j'})^\nu}{k} = \frac{(\ve_{M_{jj'}}-\ve_{j'})^\nu}{M_{jj'}} - \sum_{k=M_{jj'}+1}^d \frac{(\ve_k-\ve_{j'})^\nu}{k(k-1)}
\end{align}
where our convention is that the last sum is equal to zero for $M_{jj'}=d$. The equations \eqref{013} and \eqref{014}, together with the fact that $(\ve_{M_{jj'}}-\ve_{j'})^\nu=0$ unless $j'<j$ give, 
\begin{align}
\label{015}
\partial_\beta^\nu|_{\beta=0}F_1(\beta)  =& (-1)^\nu \sum_{j'<j} \ve_jp_{j'}  \frac{(\ve_j-\ve_{j'})^\nu}j -(-1)^\nu\sum_{jj'k=1}^d\chi\big(k>M_{jj'} \big)  \ve_j p_{j'} \frac{(\ve_k-\ve_{j'})^\nu}{k(k-1)}.
\end{align}
Next we take the derivative of \eqref{qabeta}, for $\nu\ge 1$, using that $\sum_{l=1}^{k-1}(
\ve_k-\ve_l) = (k-1)\ve_k-\sum_{l=1}^{k-1}\ve_l$,
\begin{align}
\partial^\nu_\beta|_{\beta=0} F_2(\beta) =& 
 (-1)^\nu\sum_{j<k} \ve_k p_j  \frac{(\ve_k-\ve_j)^\nu}{k}  - (-1)^\nu \sum_{jkl=1}^d\chi(j<k)\chi(l<k) \ve_l p_j \frac{(\ve_k-\ve_j)^\nu}{k(k-1)}.
\label{016}
\end{align}
Noting that $\chi(k>M_{jj'}) = \chi(j<k)\chi(j'<k)$ and doing a relabeling of the summation indices one sees that \eqref{015} equals \eqref{016}. This shows that the derivatives of all orders $\nu\ge1$ at $\beta=0$ coincide, so that the two entire functions $F_1$ and $F_2$ differ by a constant. From \eqref{qabeta} and \eqref{38} we obtain  $\lim_{\beta\rightarrow\infty}F_1(\beta)=-\sum_{k=2}^{d}\frac{p_k}{k}\sum_{l=1}^{k-1}(\ve_k-\ve_l)$ and 
\begin{align*}
\lim_{\beta\rightarrow\infty}F_2(\beta)&=\sum_{j\leq j'}\ve_j\frac{p_{j'}}{j'}-\sum_{j=1}^d\ve_jp_j=\sum_{j'=2}^d\sum_{j=1}^{j'-1}\ve_j\frac{p_{j'}}{j'}-\sum_{j=2}^d \ve_j p_j(1-\frac{1}{j})\nonumber\\
&=\sum_{k=2}^d\frac{p_{k}}{k}\sum_{l=1}^{k-1}\ve_l-\sum_{k=2}^d\frac{p_{k}}{k}\ve_k (k-1)=\sum_{k=2}^d\frac{p_{k}}{k}\sum_{l=1}^{k-1}(\ve_l-\ve_k),
\end{align*}
where we used $k-1=\sum_{l=1}^{k-1}1$. This shows that $F_1=F_2$. 
 Finally, we use \eqref{flow} and the fact that $G(\beta=0)=\frac1d \mathbf 1 \mathbf 1^T$ to get that $F_1(0)$ is given by \eqref{qa0}.  The proof of Theorem \ref{Qunbounded} is complete. \hfill\qed

\subsection{Proof of Theorem \ref{thm:QbetainR}}
\label{sec:proofthmQbar}

Our first goal is to find an expression for $Q_{\a,K}$ analogous to \eqref{flow}, \eqref{33}. Proceeding as in the proof of Proposition \ref{gjjp} we have (see \eqref{14}, \eqref{matG})
\begin{align}
\label{12K}
{\rm Tr}& \big(U(\rho\otimes\tau_{\beta,K}) U^* H_\a\big)  = \vec\ve\cdot G_K\vec p,\nonumber\\
[G_K]_{jj'} &= \sum_\ell \frac{1}{g_\ell} \sum_{k,k'=0}^K \chi\big(\ve_j+\eta_k=\ve_{j'}+\eta_{k'}=E_\ell\big)   \frac{e^{-\beta \eta_{k'}}}{Z_{\b,K}}.
\end{align}
We recall the notation $\eta_k=\frac{k}{N}$,  $E_\ell=\ve_1+\frac{\ell}{N}$ and $\ve_j=\ve_1+\frac{\ell_j}{N}$. 
The constraint $\chi(\cdots)$ on the energies in \eqref{12K} is,
\begin{align}
\ve_j+\eta_k=\ve_{j'}+\eta_{k'}=E_\ell &\Longleftrightarrow \ve_1+\frac{\ell_j}{N}+\frac kN = \ve_1+\frac{\ell_{j'}}{N}+\frac{k'}{N} = \ve_1+\frac{\ell}{N}\nonumber\\
& \Longleftrightarrow \ell_j+k=\ell_{j'}+k'=\ell.
\end{align}
For $j,j'$ fixed, the sums over  $\ell, k,k'$ in \eqref{12K} reduce to the single sum over $\ell$ in the range $\max\{\ell_j,\ell_{j'}\}\le \ell\le \min\{\ell_j,\ell_{j'}\}+K$. Therefore, as $\max\{\ell_j,\ell_{j'}\} = \ell_{\max\{j,j'\}}$ (and similar for the $\min$), and as $\eta_{k'}=E_\ell-\ve_{j'}$ in \eqref{12K}, we have
\begin{equation}
\label{12KK}
[G_K]_{jj'} = \frac{e^{\beta\ve_{j'}}}{Z_{\b,K}} \sum_{\ell=\ell_{\max\{j,j'\}}}^{\ell_{\min\{j,j'\}}+K} \frac{1}{g_\ell}  e^{-\beta E_\ell}.
\end{equation}
We assume that $K/N>\ve_d-\ve_1$. The lowest energy $E_\ell$ is $\ve_1$ and the highest one is $\ve_d+K/N$. This interval is subdivided as
$$
\ve_1 < \ve_2 <\cdots <\ve_{d-1} <\ve_d < \ve_1+\frac KN < \ve_2+\frac KN < \cdots < \ve_{d-1}+\frac KN<\ve_d +\frac KN,
$$
which corresponds to the string of inequalities for the labels of the total energy $E_\ell=\ve_1+\frac{\ell}{N}$,
$$
0=\ell_1 <\ell_2<\cdots <\ell_{d-1} <\ell_d < \ell_1+K<\ell_2+K<\cdots< \ell_{d-1}+K<\ell_d+K.
$$
Set $\ell_{d+1}=\ell_1+K,\, \ell_{d+2}=\ell_2+K,\ldots, \,\ell_{2d}=\ell_d+K$, so that the inequality string is $\ell_1<\ell_2<\cdots<\ell_{2d-1}<\ell_{2d}$. The degeneracy $g_\ell$ is constant in each subinterval: For $\ell$ in the first to the $d$-th interval, $g_\ell$ is ascending from $1$ to $d$ and after it is descending from $d-1$ back to $1$. Namely, $g_\ell=j$ for $\ell_j\le \ell<\ell_{j+1}$ if $j=1,\ldots,d-1$, then $g_\ell=d$ for $\ell_d\le\ell\le\ell_{d+1}$ and finally $g_\ell= 2d-j$ if $\ell_j< \ell\le\ell_{j+1}$ for $j=d+1,\ldots,2d-1$. 

Denote for a moment $m=\min\{j,j'\}$ and $M=\max\{j,j'\}$ (with $j,j'$ fixed). 
The degeneracies in \eqref{12KK} take the values 
\begin{equation}
g_\ell = \left\{
\begin{array}{cl}
M & \text{for $\ell_M\le \ell <\ell_{M+1}$},\\
M+1 & \text{for $\ell_{M+1}\le \ell <\ell_{M+2}$},\\
\vdots & \vdots\\
d & \text{for $\ell_d\le\ell\le\ell_{d+1}$},\\
d-1 & \text{for $\ell_{d+1}<\ell\le\ell_{d+2}$},\\
\vdots & \vdots\\
d-m+1 & \text{for $\ell_{d+m-1} < \ell\le \ell_{d+m}$.}
\end{array}
\right.
\end{equation}
Accordingly, we have 
\begin{align}
[G_K]_{jj'} = \frac{e^{\beta\ve_{j'}}}{Z_{\b,K}} &\Big[ \sum_{s=\max\{j,j'\}}^{d -1} \frac{1}{s}  \sum_{\ell_s\le \ell<\ell_{s+1}}e^{-\beta E_\ell} &\nonumber\\
&+\frac1d\sum_{\ell_d\le\ell\le\ell_{d+1}} e^{-\beta E_\ell}+ \sum_{s=1}^{\min\{j,j'\}-1} \frac{1}{d-s}  \sum_{\ell_{d+s}< \ell\le \ell_{d+s+1}} e^{-\beta E_\ell}\Big].
\label{12s}
\end{align}
We have $E_\ell=\ve_1+\ell/N$ and the sums over $\ell$ can be evaluated explicitly,
$$
\sum_{r\le \ell\le s}e^{-\beta E_\ell} = e^{-\beta \ve_1}\,  \frac{e^{-\beta r/N}-e^{-\beta(s+1)/N}}{1-e^{-\beta/N}}.
$$
Then
\begin{align}
\sum_{\ell_s\le\ell<\ell_{s+1}}e^{-\beta E_\ell} = \frac{e^{-\beta\ve_s}-e^{-\beta\ve_{s+1}}}{1-e^{-\beta/N}},&\quad \sum_{\ell_d\le \ell\le\ell_{d+1}}e^{-\beta E_\ell} =\frac{e^{-\beta\ve_d} - e^{-\beta(K+1)/N}e^{-\beta\ve_1}}{1-e^{-\beta/N}}\nonumber\\
\sum_{\ell_{d+s}<\ell\le\ell_{d+s+1}}e^{-\beta E_\ell} = & \ e^{-\beta(K+1)/N} \frac{e^{-\beta\ve_s}-e^{-\beta\ve_{s+1}}}{1-e^{-\beta/N}}
\end{align}
Using the last formulas and \eqref{f15} we obtain from \eqref{12s},
\begin{align}
\label{f19}
[G_K]_{jj'} = &
\frac{e^{\beta\ve_{j'}}}{1-e^{-\beta(K+1)/N}}
\Big[  \sum_{s=\max\{j,j'\}}^{d-1}\frac1s \big(e^{-\beta\ve_s}-e^{-\beta\ve_{s+1}}\big) \nonumber\\
& + e^{-\beta(K+1)/N} \sum_{s=1}^{\min\{j,j'\}-1}  \frac{1}{d-s}\big(e^{-\beta\ve_s}-e^{-\beta\ve_{s+1}}\big) + \frac1d \big( e^{-\beta\ve_d} - e^{-\beta(K+1)/N}e^{-\beta\ve_1}\big)\Big].
\end{align}
For $\beta=0$ we have $Z_{\b,K}=1+K$ (see also \eqref{f15}),  and using that $\sum_{r\le \ell\le s}1=s-r+1$ in \eqref{12s} we arrive at
\begin{align}
\label{beta0}
[G_K]_{jj'} = & \frac1{K+1}\Bigg[\sum_{s=\max\{j,j'\}}^{d-1}\frac{(\ve_{s+1}-\ve_s)N}{s}+\frac{K-(\ve_d-\ve_1)N}{d}+\sum_{s=1}^{\min\{j,j'\}-1}\frac{(\ve_{s+1}-\ve_s)N}{d-s}\Bigg].
\end{align}
We thus have shown that 
\begin{align}
G_{jj'}(\beta)\equiv  \lim_{K\rightarrow\infty}
& [G_K]_{jj'}(\beta)\nonumber\\[10pt]
  = & 
\left\{
\begin{array}{cl}
e^{\beta\ve_{j'}}\sum_{s=\max\{j,j'\}}^{d-1}\frac{1}{s}\big(e^{-\beta \ve_s}-e^{-\beta \ve_{s+1}}\big) + e^{\beta\ve_{j'}}\frac1d e^{-\beta\ve_d} & \text{for $\beta>0$,}\\[15pt]
-e^{\beta\ve_{j'}}  \sum_{s=1}^{\min\{j,j'\}-1}\frac{1}{d-s} \big(e^{-\beta\ve_s}-e^{-\beta\ve_{s+1}}\big)+ e^{\beta\ve_{j'}}\frac1d e^{-\beta\ve_1} & \text{for $\beta<0$,}\\[15pt]
\frac1d  & \text{for $\beta=0$,}
\end{array}
\right.
\label{f22}
\end{align}
where it is understood that the sums yield zero for $\max\{j,j'\}=d$, respectively $\min\{j,j'\}=1$.
The expression \eqref{f22} for $\beta>0$ coincides with \eqref{33} and the branches of \eqref{f22} for $\beta\neq 0$ as well as \eqref{33} equal $1/d$ for $\beta=0$. So \eqref{f22} is a continuous extension of \eqref{33} to $\beta\in\mathbb R$. We introduce the two matrices $G^\pm(\beta)$, each defined for $\beta\in\mathbb R$, by their matrix elements, (for \eqref{f24} make a change of variables in the second branch of \eqref{f22})
\begin{align}
G_{jj'}^+(\beta) & = e^{\beta\ve_{j'}}\sum_{s=\max\{j,j'\}}^{d-1}\frac{1}{s}\big(e^{-\beta \ve_s}-e^{-\beta \ve_{s+1}}\big) + e^{\beta\ve_{j'}}\frac1d e^{-\beta\ve_d},\label{f23}\\[10pt]
G_{jj'}^-(\beta) & = -e^{\beta\ve_{j'}}  \sum^{d-1}_{s=d-\min\{j,j'\}+1}\frac{1}{s} \big(e^{-\beta\ve_{d-s}}-e^{-\beta\ve_{d-s+1}}\big)+ e^{\beta\ve_{j'}}\frac1d e^{-\beta\ve_1}.
\label{f24}
\end{align}
Then by \eqref{f22}, $G(\beta)=G^+(\beta)$ for $\beta>0$ and $G(\beta)=G^-(\beta)$ for $\beta<0$. When we want to stress that $G^\pm$ are functions of the energy vector $\vec{\ve}=(\ve_1,\ldots,\ve_d)$, then we write $G^\pm(\beta,\vec{\ve}\,)$. Consider now the linear map $V$ on $\mathbb R^d\rightarrow\mathbb R^d$ defined on the canonical basis elements $\vec{e}_j$ by
\begin{equation}
\label{V}
V\vec{e}_j = \vec{e}_{d-j+1},\qquad j=1,\ldots,d.
\end{equation}
$V$ simply permutes coordinates of vectors, in particular,  $\ve_j\leftrightarrow \ve_{d-j+1}$. $V$ is unitary and $V^2=\bbbone$. Then we have 
$$
[G^-(\beta, V\vec{\ve}\,)]_{jj'} = e^{\beta\ve_{d-j'+1}}  \sum^{d-1}_{s=d-\min\{j,j'\}+1}\frac{1}{s} \big(e^{-\beta\ve_{s}}-e^{-\beta\ve_{s+1}}\big)+ e^{\beta\ve_{d-j'+1}}\frac1d e^{-\beta\ve_d}.
$$
Comparing with \eqref{f23} we get $[G^+(\beta,\vec{\ve}\,)]_{d-j+1,d-j'+1} = [G^-(\beta, V\vec{\ve}\,)]_{jj'}$, and as $V^2=\bbbone$,
\begin{equation}
\label{f28}
[G^-(\beta, \vec{\ve}\,)]_{jj'} = [G^+(\beta,V\vec{\ve}\,)]_{d-j+1,d-j'+1} = \big[V^*G^+(\beta,V\vec{\ve}\,) V\big]_{jj'}.
\end{equation}
Next we define, for $\beta\in\mathbb R$, 
\begin{equation}
\label{f25}
Q_\a^\pm(\beta) = \vec{\ve}\cdot G^\pm(\beta)\vec{p} - \vec{\ve}\cdot\vec{p}.
\end{equation}
Then as $G^-(\beta,\vec{\ve}\,)= V^* G^+(\beta,V\vec{\ve}\,)V$ by \eqref{f28} and as $V$ is unitary, we obtain
\begin{align}
Q^-_\a(\beta) = (V\vec{\ve}\,)\cdot G^+(\beta, V\vec{\ve}\,) V\vec{p} -(V\vec{\ve}\,)\cdot(V\vec{p}\,). 
\label{f30}
\end{align}
It follows from \eqref{f22}, \eqref{f23}, \eqref{f24} that  $G^\pm(\beta=0,\vec\ve\,)=\frac1d \mathbf 1\mathbf 1^T$ regardless of $\vec\ve$, so that 
\begin{equation}
\label{contflowzero}
Q^+_\a(0)=Q^-_\a(0) = \sum_{j=1}^d \ve_j(\frac1d -p_j).
\end{equation}

The right side of \eqref{f30} is given by \eqref{qabeta} in which $\vec\ve$ is replaced by 
$\vec\eta\equiv V\vec{\ve}$ and $\vec p$ is replaced by $\vec q\equiv V\vec p$,
\begin{align}
Q&^-_\a(\beta)=\sum_{1\leq j < k \leq d}\frac{e^{-\beta (\eta_k -\eta_j)}}{k(k-1)}q_j \sum_{l=1}^{k-1}(\eta_k-\eta_l)- \sum_{k=2}^{d}\frac{q_k}{k}\sum_{l=1}^{k-1}(\eta_k-\eta_l)\nonumber\\
&=\sum_{1\leq j < k \leq d}\frac{e^{-\beta (\ve_{d-k+1} -\ve_{d-j+1})}}{k(k-1)}p_{d-j+1} \sum_{l=1}^{k-1}(\ve_{d-k+1}-\ve_{d-l+1})- \sum_{k=2}^{d}\frac{p_{d-k+1}}{k}\sum_{l=1}^{k-1}(\ve_{d-k+1}-\ve_{d-l+1})\nonumber\\
&=\sum_{1\leq k < j \leq d}\frac{e^{-\beta (\ve_k -\ve_j)}}{(d-k+1)(d-k)}p_j \sum_{l=k+1}^{d}(\ve_k-\ve_l)- \sum_{k=1}^{d-1}\frac{p_k}{d-k+1}\sum_{l=k+1}^{d}(\ve_k-\ve_l).
\label{fqabeta}
\end{align}
Equivalently (as in \eqref{qabeta'}),
\begin{align}
Q^-_\a(\beta) = \sum_{1\le k< j\le d} \frac{e^{-\beta(\ve_k-\ve_j)}p_j - p_k}{(d-k+1)(d-k)} \sum_{l=k+1}^d(\ve_k-\ve_l).
\label{f31}
\end{align}
This concludes the proof of Theorem \ref{thm:QbetainR}.\hfill\qed
\medskip

For $\beta>0$, the matrix $G^+(\beta)$ coincides with $G(\beta)$ defined in \eqref{33} and we established the spectral properties of the latter in Proposition \ref{gjjp}. For completeness we address here the properties of $G^-(\beta)$ for $\beta < 0$.
\begin{lem}
For $\beta<0$ the matrix $G^-(\beta)$, \eqref{f24} has a simple eigenvalue $1$, with corresponding eigenvector $v_\beta$ (defined after \eqref{ev1}), and its spectrum is a subset of $(0,1]$.
\end{lem}
\begin{proof}
For $K\in \mathbb N$ and $\beta\in \mathbb R$ we consider the quantum channel $\Phi_{\beta, K}$ on $\mathcal{B}(\mathcal{H})$ given by
\begin{align}\label{phiK}
    \Phi_{\beta, K}(X)={\rm Tr} \big(U(X\otimes\tau_{\beta,K}) U^*\big).
\end{align}
Let $\rho$ be a density matrix (with diagonal $\vec p$\,). By \eqref{12K} we have $\tr  (\Phi_{\beta,K}(\rho)H_\a)=\vec\ve\cdot G_K(\beta)\vec p$ with 
\begin{align}\label{GKPhiK}
    [G_K(\beta)]_{jj'}= \big[\Phi_{K,\beta}(\ket{\psi_{j'}}\bra{\psi_{j'}})\big]_{jj}.
\end{align}
Note that both $\Phi_{\beta, K}$ and $G_K(\beta)$ are analytic functions of $\beta\in \mathbb C$ for finite $K$.
As $\Phi_{\beta, K}$ is CPTP and leaves the Gibbs state $\gamma_\beta$ invariant, the matrix $G_K(\beta)$ is stochastic and admits the $K$-independent vector $v_\beta$ as invariant vector, see \eqref{ev1}. Therefore, $G(\beta)=\lim_{K\rightarrow\infty}G_K(\beta)$ is stochastic for all $\beta\in \mathbb R$ and admits $v_\beta$ as invariant vector as well. Since the channel $\Phi_{K,\beta}$ is homogeneous, the stochastic matrix $G_K(\beta)$ satisfies the detailed balance relation $D^{-1}G_K(\beta)=G_K(\beta)^* D^{-1}$, where $D={\rm diag}(e^{-\beta \ve_1},\dots, e^{-\beta \ve_d})$ is independent of $K$, see \eqref{23} and Lemma \ref{lem3}. Hence, the detailed balance condition also holds for the limit stochastic matrix $G(\beta)$, which implies, as in the proof of Lemma \ref{lem3}, that the spectrum of $G(\beta)$ lies in $[-1,1]$. 

We now focus on the branch $G^-(\beta)$ for $\beta=-|\beta|<0$. As $\ve_{d-s}<\ve_{d-s+1}$ the formula \eqref{f24} gives $G^-_{j j'}(-|\beta|)>e^{|\beta|(\ve_1-\ve_{j'})}/d>0$. The Perron-Frobenius theorem then implies that $1$ is a simple eigenvalue of $G^-(-|\beta|)$ and $-1$ is not an eigenvalue. Proceeding as in the proof of Proposition \ref{gjjp}, following \eqref{yyy}, we introduce the matrix $\Gamma^-$ by $G^-(\beta)=\Gamma^- D^{-1}$ so that
\begin{align}
    \Gamma_{jj'}^-     &=  \sum^{d-1}_{s=d-\min\{j,j'\}+1}\frac{1}{s} \big(e^{|\beta|\ve_{d-s+1}}-e^{|\beta|\ve_{d-s}}\big)+ \frac1d e^{|\beta|\ve_1}.
\end{align}
By the algebraic arguments in the proof of Proposition \ref{gjjp}, the matrix $\Gamma^-$ is positive definite if and only if the spectrum of $G^-(-|\beta|)$ is strictly positive. 
Setting 
\begin{align}\label{gmoins}
    g_j^-=\begin{cases}\sum_{s=d-j+1}^{d-1} \frac{1}{s} \big(e^{|\beta|\ve_{d-s+1}}-e^{|\beta|\ve_{d-s}}\big)+\frac1d e^{|\beta|\ve_1} \ &\text{if} \ j>1 \\
    \frac1d e^{|\beta|\ve_1} \ &\text{if} \ j=1,
    \end{cases}    
\end{align}
such that $g_d^->g_{d-1}^->\dots > g_1^->0$, the matrix $\Gamma^-$ reads
\begin{align}
    \Gamma^-= \begin{pmatrix}
g_1^- & g_1^-& g_1^- & \dots & g_1^- \\
g_1^- & g_2^-  & g_2^- & \cdots & g_2^-\\
g_1^- & g_2^- & g_3^- & \cdots & g_3^- \\
\vdots & \vdots & \vdots & \ddots & \vdots\\
g_1^- & g_2^- & g_3^- & \dots & g_d^-
\end{pmatrix} = T_-P_-T_-^*, \qquad T_-=\begin{pmatrix}
1 & 0& 0 & \dots & 0 \\
1 & 1 & 0 & \cdots & 0\\
1 & 1 & 1 & \cdots & 0 \\
\vdots & \vdots & \vdots & \ddots & \vdots\\
1 & 1 & 1 & \dots & 1
\end{pmatrix},
\end{align}
and where $P_-=\text{diag}(g_1^-, g_2^--g_1^-, \dots, g_{d-1}^--g_{d-2}^-, g_d^--g_{d-1}^-)$ is a strictly positive definite matrix.
\end{proof}

\end{document}